\documentclass{article}
\pdfoutput=1

\usepackage[utf8]{inputenc}
\usepackage[margin=1in]{geometry}
\usepackage{amsthm,amsmath,amsfonts}
\usepackage{mathtools}
\usepackage{thmtools}
\usepackage{xfrac}
\usepackage{physics}
\usepackage{xcolor}
\usepackage[backref=page,colorlinks,bookmarks=true]{hyperref}
\usepackage[capitalise,noabbrev,nameinlink]{cleveref}
\usepackage[shortlabels]{enumitem}
\usepackage{svg}
\usepackage{subcaption}
\usepackage{cancel}
\usepackage{array, booktabs, makecell}

\setenumerate[1]{wide = 0pt,labelwidth = 0em, itemsep = 0em, leftmargin=*}
\newlist{defaultenumerate}{enumerate}{3}
\setlist[defaultenumerate,1]{label=\arabic*.}%
\setlist[defaultenumerate,2]{label=\arabic*.}%
\setlist[defaultenumerate,3]{label=\arabic*.}%

\renewcommand{\backref}[1]{}

\renewcommand{\backrefalt}[4]{%
\ifcase #1 %
\or
[p.\ #2]%
\else
[pp.\ #2]%
\fi}

\let\originalleft\left
\let\originalright\right
\renewcommand{\left}{\mathopen{}\mathclose\bgroup\originalleft}
\renewcommand{\right}{\aftergroup\egroup\originalright}

\makeatletter
\def\@cite#1#2{\textup{[{#1\if@tempswa , #2\fi}]}}
\makeatother

\makeatletter
\let\orgdescriptionlabel\descriptionlabel
\renewcommand*{\descriptionlabel}[1]{%
  \let\orglabel\label
  \let\label\@gobble
  \phantomsection
  \edef\@currentlabel{#1\unskip}%
  \let\label\orglabel
  \orgdescriptionlabel{#1}%
}
\makeatother

\DeclarePairedDelimiter\ceil{\lceil}{\rceil}
\DeclarePairedDelimiter\floor{\lfloor}{\rfloor}

\newcommand{\oh}[1]{O\left(#1\right)}

\theoremstyle{plain}

\newtheorem{Lemma}{Lemma}
\newtheorem*{Problem*}{Problem}
\newtheorem*{Problems*}{Problems}
\newtheorem{Theorem}{Theorem}
\newtheorem{Fact}{Fact}
\newtheorem{Corollary}{Corollary}
\newtheorem{Definition}{Definition}
\newtheorem{Proposition}{Proposition}
\theoremstyle{definition}

\newtheorem*{Remark*}{Remark}
\newtheorem*{Remarks*}{Remarks}

\makeatletter
\def\namedlabel#1#2{\begingroup
    #2%
    \def\@currentlabel{#2}%
    \phantomsection\label{#1}\endgroup
}
\makeatother

\crefname{Lemma}{lemma}{lemmas}
\crefname{Lemma}{Lemma}{Lemmas}

\crefname{Corollary}{corollary}{corollaries}
\Crefname{Corollary}{Corollary}{Corollaries}

\crefname{remark}{remark}{remarks}
\Crefname{remark}{Remark}{Remarks}

\crefname{Proposition}{proposition}{propositions}
\Crefname{Proposition}{Proposition}{Propositions}

\crefname{table}{table}{tables}
\Crefname{table}{Table}{Tables}

\crefname{description}{description}{descriptions}
\Crefname{description}{Description}{Descriptions}

\crefname{strategy}{strategy}{strategies}
\Crefname{strategy}{Strategy}{Strategies}

\crefname{property}{property}{properties}
\Crefname{property}{Property}{Properties}

\DeclareMathOperator{\adv}{\mathrm{Adv}}
\DeclareMathOperator{\gammatwo}{\gamma_2}
\DeclareMathOperator{\twodots}{..}
\DeclareMathOperator{\OR}{\mathsf{OR}}
\DeclareMathOperator{\EQUAL}{\mathsf{EQUAL}}
\DeclareMathOperator{\EXACT}{\mathsf{EXACT}_{\text{2 of 3}}}

\DeclareMathOperator{\AND}{\mathsf{AND}}
\DeclareMathOperator{\ANDOR}{\mathsf{AND-OR}}

\DeclareMathOperator{\SWITCH}{\mathsf{SWITCH-CASE}}
\DeclareMathOperator{\aux}{\mathrm{aux}}

\DeclareMathOperator{\LCS}{\mathsf{LCS}}
\DeclareMathOperator{\LIS}{\mathsf{LIS}}
\DeclareMathOperator{\ED}{\mathsf{ED}}

\newcommand{\expect}{\mathbb{E}}
\newcommand{\calO}{\mathcal{O}}
\newcommand{\calR}{\mathcal{R}}
\newcommand{\mt}{\textsf{min-last}}
\newcommand{\mh}{\textsf{max-first}}

\newcommand{\textleft}{\text{left}}
\newcommand{\textright}{\text{right}}

\newcommand{\Otilde}{\tilde{O}}
\newcommand{\LCSC}[2]{\LCS^{\mathrm{composite}}_{#1, #2}}
\newcommand{\s}{\mathbf{s}}

\title{\bfseries Quantum divide and conquer}

\author{\normalsize 
Andrew M. Childs\footnote{University of Maryland. \texttt{amchilds@umd.edu}} \;\; 
Robin Kothari\footnote{Microsoft Quantum. \texttt{robin@robinkothari.com}} \;\; 
Matt Kovacs-Deak\footnote{University of Maryland. \texttt{kovacs@umd.edu}} \;\;
Aarthi Sundaram\footnote{Microsoft Quantum. \texttt{aarthi.sundaram@microsoft.com}} \;\;
Daochen Wang\footnote{University of Maryland. \texttt{wdaochen@gmail.com}}
}

\date{\vspace{-4ex}}

\begin{document}

\maketitle

\begin{abstract}
    The divide-and-conquer framework, used extensively in classical algorithm design, recursively breaks a problem of size $n$ into smaller subproblems (say, $a$ copies of size $n/b$ each), along with some auxiliary work of cost $C^{\aux}(n)$, to give a recurrence relation
    $$C(n) \leq a \, C(n/b) + C^{\aux}(n)$$
    for the classical complexity $C(n)$.
    We describe a quantum divide-and-conquer framework that, in certain cases, yields an analogous recurrence relation
    $$C_Q(n) \leq \sqrt{a} \, C_Q(n/b) + O(C^{\aux}_Q(n))$$
    that characterizes the quantum query complexity.
    We apply this framework to obtain near-optimal quantum query complexities for various string problems, such as (i) recognizing regular languages; (ii) decision versions of String Rotation and String Suffix; and natural parameterized versions of (iii) Longest Increasing Subsequence and (iv)  Longest Common Subsequence.
\end{abstract}

\section{Introduction}
Classically, divide and conquer is a basic algorithmic framework that applies widely to many problems (see, e.g., \cite{quicksort_hoare_1962, multiplication_karatsuba_1963, fourier_cooley_tukey_1965, matrix_strassen_1969, matrix_coppersmith_winograd_1990, algorithms_cormen_2009, matrix_williams_2021}). This technique solves a problem by recursively solving smaller instances of the same problem together with some auxiliary problem.  In this paper, we develop a quantum analog of this framework that can provide quantum query complexity speedups relative to classical computation.

To motivate this investigation, consider a simple application of classical divide and conquer. Consider the problem of computing the $\OR$ function on $n$ bits. The classical deterministic query complexity of this problem is well known to be $n$, since an algorithm must read all $n$ bits in the worst case to check if any of the bits equals $1$. This upper bound is trivial since any query problem on $n$ bits can be solved by querying all $n$ bits, but to illustrate the quantum idea, we first consider re-deriving this classical upper bound using divide and conquer.

It is easy to see that for an $n$-bit string $x$, $\OR(x)=1$ if and only if either $\OR(x_\textrm{left})=1$ or $\OR(x_\textrm{right})=1$, where $x_\textrm{left}$ and $x_\textrm{right}$ are the left and right halves of $x$. Thus we have divided the original problem into two smaller instances of the same problem, and given solutions to the two smaller instances, no further queries need to be made to solve the larger instance. Letting $C(n)$ denote the classical query complexity of this problem, this argument yields the recurrence relation
\begin{equation}
    C(n) \leq 2\, C(n/2),
\end{equation}
which, noting $C(1) = 1$, solves to $C(n)\leq n$, as desired.

Now notice that the last step that combines the solutions of the smaller instances to solve the larger instance is also an $\OR$ function, but on $2$ bits, since $\OR(x)=\OR(x_\textrm{left}) \vee \OR(x_\textrm{right})$. We know that the quantum query complexity of the $\OR$ function on $n$ bits is $O(\sqrt{n})$ due to Grover's algorithm \cite{search_grover_1996}, so we might be tempted to say that the $\OR$ of 2 bits should be faster to compute, potentially even quadratically faster. Thus we might like to write the following recurrence relation for the quantum query complexity $Q(n)$ of computing the $\OR$ function:
\begin{equation}\label{eq:scarequotes}
    Q(n) \leq \text{``$\sqrt{2}$''} \, Q(n/2).
\end{equation}
Here the quotes around the ``$\sqrt{2}$'' convey that this is not really justified. First, the complexity of Grover's algorithm is not exactly $\sqrt{n}$, but only $\sqrt{n}$ up to a constant factor. Even if it were exactly $\sqrt{n}$, Grover's algorithm is a bounded-error algorithm, so if used as a subroutine that calls itself many times, the error will become unbounded very quickly. And of course, it does not make sense to imagine that the query complexity is a non-integer.

Even though we have just argued that \Cref{eq:scarequotes} is questionable, if we solve this recurrence anyway, we find $Q(n) \leq \sqrt{n}$.\footnote{This may be familiar to readers who know the quantum adversary method. However, our divide-and-conquer framework generalizes this phenomenon to a scenario that cannot be captured using adversary machinery alone, as discussed below.} Up to a multiplicative constant, this is the correct answer for the quantum query complexity of the $\OR$ function on $n$ bits! 

This is not a coincidence. As we show in this paper, this can be seen as a simple instance of a quantum algorithmic framework that we call ``quantum divide and conquer.'' More generally, a classical divide-and-conquer algorithm typically yields a recurrence relation of the form
\begin{equation}\label{eq:classical_recurrence}
    C(n) \leq a\,C(n/b) + C^{\aux}(n),
\end{equation}
where $C^{\aux}(n)$ is the classical query complexity of the auxiliary problem.

The main conceptual takeaway of our paper is that, in many settings, the quantum query complexity equals (up to a multiplicative constant) the solution, $C_Q(n)$, 
of the analogous recurrence relation,
\begin{equation}\label{eq:quantum_recurrence}
    C_Q(n) \leq \sqrt{a} \, C_Q(n/b) + C^{\aux}_Q(n),
\end{equation}
where the quantum query complexity of the auxiliary problem is $O(C^{\aux}_Q(n))$. We show that this framework easily recovers and improves non-trivial quantum speedups that were previously known. Furthermore, we show it yields previously unknown quantum speedups that do not seem to be achievable by standard applications of techniques such as Grover search~\cite{search_grover_1996}, amplitude amplification and estimation~\cite{amplitude_brassard_2002}, and quantum walk search~\cite{qwalk_elem_distictness_ambainis_2007}. 

\subsection{Results}

We apply quantum divide and conquer to the following decision problems, which we have ordered in roughly increasing difficulty. In all of these  problems, we are given query access to a string $s\in\Sigma^a$ over some set $\Sigma$, i.e., query access to a function $f\colon \{1,\ldots,a\} \to \Sigma$, with $f(i) = s_i$.

\paragraph{Recognizing regular languages.} We recover the key algorithmic result in \cite{trichotomy_aaronson_2019} that is used to upper bound the quantum query complexity of recognizing star-free languages. (The latter result is, in turn, the key result used by \cite{trichotomy_aaronson_2019} to establish a quantum query complexity trichotomy for recognizing regular languages.) In particular, we show that $O(\sqrt{n\log n})$ quantum queries suffice to decide whether a string $x \in \{0,1,2\}^n$ contains a substring of the form $20^*2$, i.e., two $2$s on either side of an arbitrarily long string of $0$s. Having established our framework, the analysis is simpler than that of \cite{trichotomy_aaronson_2019}.

\paragraph{String Rotation and String Suffix.} Let $x\in \Sigma^n$ and $i\in \{1,\ldots,n\}$. In String Rotation, the problem is to decide whether the lexicographically minimal rotation of $x$ starts at $i$. In String Suffix, the problem is to decide whether the lexicographically minimal suffix of $x$ starts at $i$.
These problems are natural decision versions of problems studied in \cite{string_akmaljin_2022}. We obtain an $\tilde{O}(\sqrt{n})$ upper bound on the quantum query complexity for these problems, improving the previous best bound of $O(n^{1/2+o(1)})$, which follows from the algorithms of \cite{string_akmaljin_2022} for the search versions. Furthermore, the analysis is again simpler using our framework.

\paragraph{$k$-Increasing Subsequence ($k$-IS).} Given a string $x\in \Sigma^n$ over some ordered set $\Sigma$ and an integer $k\geq 1$, the $k$-IS problem asks us to decide whether $x$ has an increasing subsequence\footnote{Recall that a subsequence of a string is obtained by taking a subset of the elements without changing their order. Note that this is more general than a sub\emph{string}, in which the chosen elements must be consecutive.} of length at least $k$. We obtain a bound of $\Otilde(\sqrt{n})$ for any constant $k$. The $k$-IS problem is a natural parameterized version of Longest Increasing Subsequence (LIS), which is well-studied classically (see, e.g.,~\cite{lis_fredman_1975, lis_aldous_1999,lis_saks_2010,lcs_lis_rubinstein_2019}). We consider this version because the quantum query complexity of LIS is easily seen to be\footnote{This follows by reduction from a threshold function such as majority \cite{lower_bbcmw_2001}. Given a string $z \in \{0,1\}^n$, define $x \in \{0,1,\ldots,n\}^n$ such that $x_i = i z_i$ for all $i \in \{1,\ldots,n\}$. Then we can determine the Hamming weight $|z|$ of $z$ by classically querying $z_1$ and adding $z_1-1$ to the length of the LIS of $x$.} $\Omega(n)$, so we cannot hope for a quantum speedup without a bound on $k$. Somewhat surprisingly, any constant bound on $k$ allows for a quadratic speedup over classical algorithms.

\paragraph{$k$-Common Subsequence ($k$-CS).} Given two strings $x,y\in \Sigma^n$  over some set $\Sigma$ and integer $k\geq 1$, the $k$-CS problem asks us to decide whether $x$ and $y$ share a common subsequence of length at least $k$. We obtain a bound of $\Otilde(n^{2/3})$ for any constant $k$. The $k$-CS problem is a natural parametrized version of Longest Common Subsequence (LCS), which is well-studied classically (see, e.g., \cite{lcs_wagner_1974,lcs_landau_1998,algorithms_cormen_2009,lcs_asymmetric_andoni_2010,lcs_abboud_2015, lcs_lis_rubinstein_2019}). Again, we consider this version because the quantum query complexity of LCS is easily seen to be\footnote{This also follows by reduction from a threshold function such as majority. Given a string $z\in \{0,1\}^n$, define $x=1^n$. Then we can determine the Hamming weight $|z|$ of $z$ as the length of the longest common subsequence of $z$ and $x$.} $\Omega(n)$.

Note that LCS is closely connected to Edit Distance (ED), which asks us to compute the number of insertions, deletions, and substitutions required to convert $x$ into $y$. In fact, if substitutions were disallowed, the edit distance between $x$ and $y$ would equal $\ED(x,y) = 2n - 2\, \LCS(x,y)$ (a proof can be found in, e.g., \cite[Lemma 6]{lcs_akmal_2021}).

\paragraph{}All of our quantum query complexity upper bounds are tight up to logarithmic factors and represent quantum-over-classical speedups since the classical (randomized) query complexities of all problems considered are $\Omega(n)$.

\subsection{Techniques}
\label{sec:techniques}
To derive the quantum recurrence relation in \Cref{eq:quantum_recurrence}, we employ the quantum adversary method as an upper bound on quantum query complexity~\cite{reflections_reichardt_2011,stateconversion_lee_2011}. The adversary quantity $\adv(P(n))$ is a real number that can be associated with any query problem $P(n)$ of size $n$. For convenience, we write $\adv(n) \coloneqq \adv(P(n))$ when $P$ is clear from context. It is well-known that $\adv(n)$ is upper and lower bounded by the (two-sided bounded error) quantum query complexity, $Q(n)$, of $P(n)$ up to multiplicative constants. 

If a problem can be expressed as the composition of smaller problems, then the adversary quantity of the larger problem can be expressed in terms of those of the smaller problems~\cite{reflections_reichardt_2011,stateconversion_lee_2011}.
For example, suppose $P(n)$ is the $\OR$ of $a$ copies of $P(n/b)$ together with another problem $P^{\aux}(n)$, corresponding to the auxiliary work performed in a divide-and-conquer algorithm. Then the composition theorem implies
\begin{equation}
    \adv(n)^2 \leq a\, \adv(n/b)^2 + \adv(P^{\aux}(n))^2.
\end{equation}
(For a more precise statement, see~\Cref{sec:framework}.)
Taking square roots\footnote{Note that while taking square roots gives a closer analogy with the classical recurrence, in some cases we can get a slightly tighter bound by instead directly solving the recurrence for $\adv(n)^2$.} and using the fact that $\adv(P^{\aux}(n))$ is bounded by $O(Q^{\aux}(n))$, where $Q^{\aux}(n) \coloneqq Q(P^{\aux}(n))$, we obtain precisely \Cref{eq:quantum_recurrence}, i.e.,
\begin{equation}\label{eq:quantum_recurrence_copy}
    \adv(n) \leq \sqrt{a}\, \adv(n/b) + O(Q^{\aux}(n)).
\end{equation}
We then bound $Q^{\aux}(n)$ by constructing an explicit quantum algorithm for the auxiliary work and bounding its quantum query complexity. Finally, we solve \Cref{eq:quantum_recurrence_copy} either directly or by appealing to the Master Theorem~\cite{master_theorem_bentley_haken_saxe_1980} to obtain a bound on $\adv(n)$ and hence a bound on $Q(n)$. Note that this bound is insensitive (up to a multiplicative constant) to the constant implicit in $O(Q^{\aux}(n))$, whereas it is highly sensitive to the constant in front of $\adv(n/b)$.

We emphasize that the recurrence in \Cref{eq:quantum_recurrence_copy} involves \emph{both} the adversary quantity and the quantum query complexity. At a high level, this is why the bound on $Q(n)$ resulting from solving \Cref{eq:quantum_recurrence_copy} does not follow directly from adversary techniques alone nor from quantum algorithmic techniques alone. Indeed, when we solve \Cref{eq:quantum_recurrence_copy} as a recurrence relation, we mix adversary and algorithmic techniques at each step of the recursion. In principle, it is possible to use \Cref{eq:quantum_recurrence_copy} to construct an explicit quantum algorithm for $P(n)$ because there is a constructive, complexity-preserving, and two-way correspondence between span programs---whose complexity is characterized by the adversary quantity---and quantum algorithms~\cite{span_reichardt_2009,span_cornelissen_2020, span_jeffery_2022}. However, this correspondence is involved.

This quantum divide-and-conquer framework allows us to use \emph{classical} divide-and-conquer thinking to come up with a classical recurrence relation that we can easily convert to a corresponding quantum recurrence relation in the form of \Cref{eq:quantum_recurrence_copy}. Bounding $Q(n)$ then reduces to bounding $Q^{\aux}(n)$, which can be easier than the original problem. In our applications to $k$-IS and $k$-CS, we bound $Q^{\aux}(n)$ by using quantum algorithms for the $i<k$ versions of these problems, whose query complexities we can bound inductively. Note that the base cases ($k=1$) of these problems are trivial, being equivalent to search and (bipartite) element distinctness, respectively.

The form of \Cref{eq:quantum_recurrence_copy} depends on the way we divide and conquer. Strikingly, in our application to $k$-CS, we find that an optimal (up to logarithmic factors) quantum query complexity can be derived by breaking the problem into \emph{seven} parts (so that $b=7$) but not fewer. Classically, the number of parts we break the problem into makes no difference as any choice leads to an $O(n)$ classical query complexity.

In the discussion above, we have assumed that the $\OR$ function relates $P(n)$ to $P(n/b)$ and $P^{\aux}(n)$. However, this is only to simplify our exposition. In fact, we can use quantum divide and conquer for \emph{any} function relating $P(n)$ to $P(n/b)$ and $P^{\aux}(n)$. Functions that we use to obtain quantum speedups include combinations of $\OR$ and $\AND$. In our application of quantum divide and conquer to $k$-CS, the function is a combination of $\SWITCH$ and $\OR$. While $\SWITCH$ alone yields no quantum speedup, our analysis relies on an adversary composition theorem that we prove, which allows us to preserve the speedup yielded by $\OR$. We remark that there are  many other functions that could yield quantum speedups, but which we have not yet considered in applications. Examples include $\EQUAL$ (which decides if the subproblems all return the same answer or not) and $\EXACT$ (which decides if exactly $2$ of $3$ subproblems return $1$)---see \cite[Table 1]{formulas_reichardt_2012}---and their combinations with $\OR$, $\AND$, and $\SWITCH$. We leave the consideration of such functions for future work.

\subsection{Related work}

The technique of using adversary composition theorems (and their relatives) to establish upper bounds on quantum query complexity has been applied in several previous works, beginning with the setting of formula evaluation~\cite{unbalanced_reichardt_2009, formulas_reichardt_2012}. In \cite{adversaryupper_kimmel_2013}, an adversary composition theorem is used to bound the quantum query complexity of a function $f$ given a bound on that of $f$ composed with itself $d$ times. More generally, the adversary quantity has been used to upper bound the quantum query complexity of problems including $k$-distinctness~\cite{kdist_belovs_2012}, triangle finding~\cite{span1_belovs_2012,triangle_lee_2017}, undirected $st$-connectivity~\cite{st_belovs_2012}, and learning symmetric juntas~\cite{juntas_belovs_2014}. In contrast to our work, these prior results typically bound the adversary quantity using only tools from the adversary world, such as span programs, semidefinite programming, and composition theorems, without constructing quantum algorithms to use as subroutines.

The first two application areas of quantum divide and conquer that we consider (recognizing regular languages and String Rotation and String Suffix) arise in \cite{trichotomy_aaronson_2019,string_akmaljin_2022} as discussed above. Turning to $k$-IS and $k$-CS, we first note that, despite the importance of these problems, there have been no significant works on their quantum complexities. Directly using Grover search~\cite{search_grover_1996} over all possible positions of a $1$-witness gives trivial $O(n)$ bounds as soon as $k\geq 2$. Directly using quantum walk search~\cite{qwalk_elem_distictness_ambainis_2007} also yields highly sub-optimal bounds of $O(n^{k/(k+1)})$ for $k$-IS and $O(n^{2k/(2k+1)})$ for $k$-CS.

Classically, LIS and LCS are typically solved using dynamic programming. There are many recent works on quantum speedups for dynamic programming, e.g., \cite{dp_ambainis_2019, dp_abboud_2019, dp_glos_2021,2d_grid_ambainis_2020, dp_klevickis_2022}. The main idea of these works is to \emph{classically} query a part of the input and to repeatedly use the result together with a quantum algorithm to solve many overlapping subproblems. However, we found it difficult to apply this idea to $k$-IS and $k$-CS.

Given the generality of the results in \cite{trichotomy_aaronson_2019} and the fact that we can recover some of its results, one might ask if the converse holds, i.e., if \cite{trichotomy_aaronson_2019} can recover our results on $k$-IS and $k$-CS. Indeed, the results of \cite{trichotomy_aaronson_2019} do apply\footnote{For example, $k$-IS can be expressed in first-order logic as $\bigvee_{a_1 < a_2 < \cdots < a_k} \exists i_1, \ldots, i_k (i_1 < \cdots < i_k) \bigwedge_{j=1}^k \pi_{a_j} (i_j)$ where $\pi_a(i)$ indicates that the $i$th element is $a$. Expressibility in first-order logic is equivalent to being in a star-free language, which can be recognized with quantum query complexity $\tilde{\Theta}(\sqrt{n})$ by \cite{trichotomy_aaronson_2019}.} to $k$-IS and $k$-CS when the set size $|\Sigma|$ is \emph{constant}. However, under that assumption, these problems can be easily solved by Grover search~\cite{search_grover_1996} and quantum minimum finding~\cite{minimum_finding_durr_hoyer_1996} using $O(\sqrt{n})$ queries. Moreover, $k$-CS becomes qualitatively easier when $|\Sigma|$ is constant as its complexity drops to $O(\sqrt{n})$, down from the $\tilde{\Theta}(n^{2/3})$ bound that we establish when $|\Sigma|$ is unbounded. Note that $\tilde{\Theta}(n^{2/3})$ is not in the trichotomy of \cite{trichotomy_aaronson_2019}, i.e., not one of $\Theta(1)$, $\tilde{\Theta}(\sqrt{n})$, or $\Theta(n)$.

\paragraph{Organization.} This paper is organized as follows. We first introduce background material and describe how to use the quantum divide-and-conquer framework for certain scenarios in~\Cref{sec:framework}. Then we apply the framework to four string problems in~\Cref{sec:applns}, obtaining near-optimal quantum query complexities for each. In particular, we consider recognizing regular languages in \Cref{subsec:regular}, String Rotation and String Suffix in \Cref{subsec:stringrot}, $k$-Increasing Subsequence in \Cref{subsec:kis}, and $k$-Common Subsequence in \Cref{subsec:kcs}. Finally, we conclude by presenting some open problems in~\Cref{sec:future}.

\section{Framework}\label{sec:framework}

In this section we introduce some notation and explain how to apply the quantum divide-and-conquer framework in specific scenarios.

Let $\mathbb{N}$ denote the positive integers. We reserve $m, n \in \mathbb{N}$ for positive integers and $\Sigma$ for a finite set. For $m\in \mathbb{N}$, we write $[m]\coloneqq \{1,\ldots,m\}$. For any two matrices $A$ and $B$ of the same size, $A \circ B$ denotes their entrywise product, also known as their Hadamard product. For any matrix $A$, let $(A)_{ij}$ denote the element in the $i$th row and $j$th column. For a function $f\colon D \to E$, let its Gram matrix $F$ be defined as $(F)_{xy} \coloneqq \delta_{f(x),f(y)}$ for all $x,y\in D$, where $\delta$ is the Kronecker delta function. For a length-$n$ string $x \in \Sigma^n$, we write $x[i]$ for the $i$th element of $x$ and $x[i \twodots j]$ for the substring of $x$ that ranges from its $i$th to $j$th elements (inclusive). We write $x(i\twodots j] \coloneqq x[i+1\twodots j]$ and $x[i\twodots j) \coloneqq x[i\twodots j-1]$. When $\Sigma$ is equipped with a total order, and $x\in \Sigma^m$ and $y\in \Sigma^n$, we write inequalities between $x$ and $y$ (e.g., $x\leq y$) to refer to inequalities with respect to the lexicographic order.

We write $Q(f)$ for $Q_{1/3}(f)$ (the quantum query complexity of $f$ with failure probability at most $1/3$).

A closely related quantity that is also used widely, and serves as both a lower and upper bound for the quantum query complexity of function computation, is the adversary quantity, $\adv (\cdot)$.

\begin{Definition}[Adversary quantity]\label{def:adv}
Let $D$, $E$, and $\Sigma$ be finite sets and $n\in \mathbb{N}$ with $D \subseteq \Sigma^n$. Let $f\colon D \to E$ be a function with Gram matrix $F$ and $\Delta = \{\Delta_1, \ldots, \Delta_n\}$ with $(\Delta_j)_{x, y \in D} = 1 - \delta_{x_j,y_j}$. Then the \emph{adversary quantity} for $f$ is
\begin{align*}
    \adv(f) := \max  \Big\{ \norm{\Gamma} : \forall j \in [n], \, \norm{\Gamma \circ \Delta_j} \leq 1 \Big\},
\end{align*}
where the maximization is over $|D| \times |D|$ real, symmetric matrices $\Gamma$ satisfying $\Gamma \circ F = 0$.
\end{Definition}

The following result connecting $\adv(\cdot)$ and $Q(\cdot)$ will be useful.

\begin{Theorem}[\cite{adversary_hoyer_2007,stateconversion_lee_2011}]\label{thm:adv-q}
Let $f\colon D \to E$ be a function. Then $Q(f) = \Theta(\adv(f))$.
\end{Theorem}

The adversary quantity for a composite function can be expressed in terms of the adversary quantities of its components. In this work, we specifically consider cases where (i) $g = f_1 \vee f_2$; (ii) $g = f_1 \wedge f_2$; or (iii) a $\SWITCH$ scenario on functions $f\colon D \to E$ and $\{g_s : D \to \{0, 1\} \mid s \in E\}$ defined as $h \coloneqq (\text{IF } f(x) = s) \text{ THEN } h(x) = g_{s}(x)$ for $s \in E$, $x \in D$.

\begin{restatable}[{Adversary composition for $\OR$ and $\AND$}]{Lemma}{lemAdversaryAndOr}
\label{lem:and-or}
Let $a,b \in \mathbb{N}$ and let $\Sigma$ be a finite set. Let $f_1 \colon \Sigma^a \to \{0,1\}$ and $f_2 \colon \Sigma^b \to \{0,1\}$. Let $g^\wedge, g^\vee \colon \Sigma^a \times \Sigma^b \to \{0,1\}$ be such that $g^\wedge(x,y) = f_1(x) \wedge f_2(y)$ and $g^\vee(x,y) = f_1(x) \vee f_2(y)$. Then $\adv(g^\wedge)^2 = \adv(g^\vee)^2 \leq \adv(f_1)^2 + \adv(f_2)^2$.
\end{restatable}

\Cref{lem:and-or} follows from \cite{stateconversion_lee_2011_v1,reflections_reichardt_2011}. For completeness, we provide a proof in \Cref{app:adv-composition}. Note that the lemma still holds if $f_1$ and $f_2$ share variables.\footnote{For example, if $h\colon \Sigma^a \to \{0,1\}$ with $h(x) \coloneqq g(x,x) = f_1(x) \wedge f_2(x)$, then $\adv(h)^2 \leq \adv(g)^2 \leq \adv(f_1)^2 + \adv(f_2)^2$.}

\begin{restatable}[Adversary composition for $\SWITCH$]{Lemma}{lemAdversarySwitch}\label{lem:adv-switch}
Let $a \in \mathbb{N}$ and let $\Sigma, \Lambda$ be finite sets. Let $f \colon \Sigma^a \to \Lambda$. 
For $s\in \Lambda$, let $g_s \colon \Sigma^a \to \{0,1\}$. Define $h \colon \Sigma^a \to \{0,1\}$ by $h(x) = g_{f(x)}(x)$. Then
\begin{equation}
    \adv(h) \leq O(\adv(f)) + \max_{s\in \Lambda} \adv(g_s).
\end{equation}
\end{restatable}

See~\Cref{app:adv-composition} for a proof of \Cref{lem:adv-switch}.

A divide-and-conquer strategy for a problem of size $n$ typically involves dividing the problem into $a$ subproblems on smaller instances of size $n/b$ for $b > 1$ and combining the solutions by solving another auxiliary problem $P^{\aux}$. The classical cost for $P^{\aux}(n)$ is denoted $C^{\aux}(n)$. The classical cost of solving $P(n)$ is then described by the recurrence
\begin{align}
\label{eq:class_rec}
    C(n) \leq a \, C(n/b) + C^{\aux}(n).
\end{align}

The situation is more subtle in the quantum case. \Cref{eq:class_rec} bounds the cost of solving the size-$n$ problem by summing the costs of solving each subproblem and the auxiliary problem. However, our quantum recurrence relations come from adversary composition theorems, which impose restrictions on how the size-$n$ problem decomposes into subproblems and an auxiliary problem. 

We consider decision problems that correspond to the computation of a Boolean function $f\colon \Sigma^n \to \{0, 1\}$, where $\Sigma$ is a finite set, and we have oracle access to an input $x\in \Sigma^n$. While the most general form of quantum divide and conquer can compute $f$ by breaking it into sub-functions and an auxiliary function (corresponding to the subproblems and an auxiliary problem) in many different ways, here we focus on two strategies.

\begin{enumerate}[label=\textbf{Strategy \arabic*.},ref=\arabic*,leftmargin=*]
    \item \label[strategy]{itm:strat1} $f$ is computed as an $\ANDOR$ formula involving an auxiliary function $f^{\aux}\colon \Sigma^n \to \{0,1\}$ and sub-functions $\{f_{i}\colon \Sigma^{n/b} \to \{0,1\} \mid i\in [a]\}$, where $a, b \in \mathbb{N}$. 
    \item \label[strategy]{itm:strat2} $f$ is computed sequentially as follows: (1) Compute an auxiliary function $f^{\aux}\colon \Sigma^n \to \Lambda$, where $\Lambda$ is a finite set, to obtain some $s \in \Lambda$; (2) $s$ indicates a set of sub-functions $\bigl\{f_i^{(s)}: \Sigma^{n/b}\to \{0,1\} \mid i \in [a]\bigr\}$, where $a,b\in \mathbb{N}$, and a Boolean function $g^{(s)}\colon \{0,1\}^a \to \{0,1\}$, such that $f$ is computed as the function $g^{(s)} \circ \bigl(f_1^{(s)}, \ldots, f_a^{(s)}\bigr)\colon (\Sigma^{n/b})^a \to \{0,1\}$ defined by 
    \begin{equation}
        g^{(s)} \circ \bigl(f_1^{(s)}, \ldots, f_a^{(s)}\bigr)(x^{(1)}, \ldots, x^{(a)}) \coloneqq g^{(s)}(f_1^{(s)}(x^{(1)}), \ldots, f_a^{(s)}(x^{(a)})).
    \end{equation}
\end{enumerate}

Using the adversary composition theorems stated previously, the next corollary shows how to obtain quantum recurrence relations for $\adv(f)$ when $f$ is computed  according to \Cref{itm:strat1} or \Cref{itm:strat2}.

\begin{Corollary}[Quantum divide-and-conquer strategies]\label{cor:strategies}
Let $f$, $f^{\aux}$, $\{f_i \mid i \in [a]\}$, $\bigl\{f_i^{(s)} \mid i \in [a], \, s\in \Lambda\bigr\}$, and $\{g^{(s)} \mid s\in \Lambda\}$ be as defined in \Cref{itm:strat1,itm:strat2}. Then
    \begin{alignat}{3}
        &\adv(f)^2 &&\leq \sum_{i=1}^a \adv(f_i)^2 + O\left(Q(f^{\aux})\right)^2 &&\quad \text{for \Cref{itm:strat1}}, \label{eq:strat1}\\
        &\adv(f) &&\leq O\bigl(Q(f^{\aux})\bigr) + \max_{s \in \Lambda} \adv\bigl( g^{(s)} \circ \bigl(f_1^{(s)}, \ldots, f_a^{(s)}\bigr)\bigr) &&\quad \text{for \Cref{itm:strat2}}. \label{eq:strat2}
    \end{alignat}
\end{Corollary}

\begin{proof}
For~\Cref{itm:strat1}, $f^{\aux}$ and $\{f_i \mid i\in [a]\}$ are all Boolean functions and $f$ is computed as an $\ANDOR$ formula of these functions. Therefore, we can directly apply \Cref{lem:and-or} to obtain
\begin{align}
\begin{aligned}
    \adv(f)^2 &\leq \adv(f_1)^2 + \cdots + \adv(f_a)^2 + \adv(f^{\aux})^2 \\
    &\leq \adv(f_1)^2 + \cdots + \adv(f_a)^2 + O(Q(f^{\aux}))^2 && \text{(using \Cref{thm:adv-q})}.
\end{aligned}
\end{align}

For~\Cref{itm:strat2}, $f^{\aux}\colon \Sigma^n \to \Lambda$ is first computed to give an $s\in \Lambda$. Then, $f$ is computed as $g^{(s)} \circ \bigl(f_1^{(s)} \ldots, f_a^{(s)}\bigr)$. Therefore, $f(x)$ can be expressed using $\SWITCH$ as
\begin{equation}
    f(x) = h_{f^{\aux}(x)}(x), \quad \text{where
    $h_s \coloneqq \bigl(g^{(s)} \circ \bigl(f_1^{(s)}, \ldots, f_a^{(s)}\bigr)\bigr)$ for all $s\in \Lambda$}.
\end{equation}
Therefore, we can directly apply \Cref{lem:adv-switch} to obtain
\begin{align}
    \adv(f) & \leq O(\adv(f^{\aux})) + \max_{s \in \Lambda} \adv\bigl(g^{(s)} \circ \bigl(f_1^{(s)}, \ldots, f_a^{(s)}\bigr)\bigr) \\
    & \leq O(Q(f^{\aux})) + \max_{s \in \Lambda} \adv\bigl(g^{(s)} \circ \bigl(f_1^{(s)}, \ldots, f_a^{(s)}\bigr)\bigr) &&\text{(using \Cref{thm:adv-q})}.
\end{align}
The corollary follows.
\end{proof}

While adversary composition theorems play a key role in decomposing the adversary quantity for $f$, one novel aspect of our framework lies in the switching from the adversary quantity $\adv(f^{\aux})$ to the quantum query complexity $Q(f^{\aux})$ when handling the auxiliary function. This switching means our framework employs both the adversary and quantum algorithms toolboxes.

We recall a few well-known quantum query complexity bounds that are used in subsequent sections, usually to bound the quantum query complexity of the auxiliary function $f^{\aux}$.

\begin{Theorem}[Quantum query complexity bounds]\label{thm:algorithms}
\label{thm:query-bounds}
Let $m, n \in \mathbb{N}$ with $m \leq n$ and consider functions $f\colon \Sigma \to \{0, 1\}$ and $g\colon D \to E$ for finite sets $\Sigma, D, E$.
\begin{itemize}[itemsep=0pt]
\item \emph{Unstructured search.} For the function $f$ with $|\Sigma| = n$, finding an $x$ such that  $f(x) = 1$ has quantum query complexity $O(\sqrt{n})$ \cite{search_grover_1996}.
\item \emph{Minimum finding.} Given an unsorted list $x = \{x_1, \ldots, x_n\}$, the quantum query complexity of finding the minimum (or maximum) element in the list is $O(\sqrt{n})$ \cite{minimum_finding_durr_hoyer_1996}.
\item \emph{String matching.} Given two strings $x, y$ with $|x| = n, |y| = m$, determining if $y$ is a substring of $x$ has quantum query complexity $\tilde{O}(\sqrt{m} + \sqrt{n})$ \cite{exact_string_matching_ramesh_vinay_2003}.
\item \emph{String comparison.} Given two strings $x, y \in \Sigma^n$ for an ordered set $\Sigma$, determining whether $x < y$ has quantum query complexity $O(\sqrt{n})$. This is because the problem reduces to minimum finding.
\item \emph{Bipartite element distinctness.} Given two strings $x, y$ such that $|x| = n, |y| = m$, and $m\leq n$, finding $i \in [n]$ and $j \in [m]$ such that $x_i = y_j$ has quantum query complexity $O(n^{2/3})$ \cite{qwalk_elem_distictness_ambainis_2007}. 
\end{itemize}
\end{Theorem}

Finally, once we have the quantum recurrence relation for a problem along with a bound on the quantum query complexity of $f^{\aux}$, we use the Master Theorem to solve the recurrence to bound $\adv(f)$, and thereby $Q(f)$. This is just like in classical divide and conquer. 

We state the Master Theorem below for convenience.
\begin{Theorem}[Master Theorem \cite{master_theorem_bentley_haken_saxe_1980}]\label{lemma:master_theorem}
Let $A\colon \mathbb{N} \to \mathbb{R}$, $a,b,c>0$, and $p\geq 0$. Suppose $A$ satisfies the recurrence
\begin{equation}
    A(n) = a \, A(n/b) + O(n^c\log^{p} n).
\end{equation}
Then
\begin{align}
    A(n) = \begin{cases}
  O(n^{\log_b{a}}) & \text{if $\log_b{a} > c$}, \\
  O(n^{\log_b{a}} \log^{p+1} n) & \text{if $\log_b{a} = c$}, \\ 
  O(n^{c} \log^{p} n) & \text{if $\log_b{a} < c$}.
\end{cases}
\end{align}
\end{Theorem}

\section{Applications}
\label{sec:applns}

In this section we apply the divide-and-conquer strategies described previously to various string problems. Note that we divide a problem on strings of length $n$ into subproblems on strings of length $n/b$ for $b > 1$. Without loss of generality, we assume that $n = b^r$ for some $r \in \mathbb{N}$ so that $n/b^k \in \mathbb{N}$ for all $k\in [r]$.\footnote{When $n\neq b^r$, we can divide into subproblems of sizes $\ceil{n/b}$ or $\floor{n/b}$ without affecting the validity of our arguments.}

\subsection{Recognizing regular languages}
\label{subsec:regular}

We first consider the problem of recognizing regular languages to demonstrate how quantum divide and conquer allows one to prove \cite[Theorem 18]{trichotomy_aaronson_2019}, without needing the intricacies of the original proof. In the following, we prove a special case of this result, mentioned in the abstract of \cite{trichotomy_aaronson_2019}, that captures its essence. We do not prove \cite[Theorem 18]{trichotomy_aaronson_2019} in its full generality only because that would require introducing significant terminology related to regular languages but would not offer further insight. Let $\Sigma = \{0,1,2\}$.
\begin{Problem*}[Recognizing $\Sigma^*20^*2\Sigma^*$]
Given query access to a string $x\in \Sigma^n$, decide if $x\in \Sigma^*20^*2\Sigma^*$, i.e., if $x$ contains a substring with two $2$s on either side of an arbitrarily long string of $0$s.
\end{Problem*}

\begin{Theorem}
The quantum query complexity of Recognizing $\Sigma^*20^*2\Sigma^*$ is $O(\sqrt{n\log(n)})$.
\end{Theorem}

\begin{proof}
We upper bound the quantum query complexity of the function $f_n\colon \{0,1,2\}^n \to \{0,1\}$ defined by $f_n(x) = 1$ iff $x\in \Sigma^* 2 0^* 2 \Sigma^*$. 

Clearly, $x$ contains the expression $2 0^* 2$ if and only if the expression is contained (i) entirely in the left half of the string; or (ii) entirely in the right half of the string; or (iii) it is partly contained in the left half and partly in the right half of the string. Therefore, we have
\begin{equation}\label{eq:regular_recurrence}
    f_n(x) = f_{n/2}(x_\textleft) \vee f_{n/2}(x_\textright) \vee g_n(x),
\end{equation}
where $g_n\colon \Sigma^n \to \{0,1\}$ is defined by
\begin{equation}
    g_n(x) = \begin{cases}
    1 & \text{if $x[i \twodots j]$ is of the form $20^*2$ for some $i \leq n/2 < j$,}\\
    0 & \text{otherwise.}
    \end{cases}
\end{equation}

Clearly, $g_n$ can be decided by the following algorithm using $O(\sqrt{n})$ queries:
\begin{defaultenumerate}
    \item Grover search for the maximal index $i \leq n/2$ such that $x_i=2$ and the minimal index $j > n/2$ such that $x_j = 2$. If either fails to exist, output $0$.
    \item Decide whether $x[i+1 \twodots j-1]$ equals $0^{j-i-1}$ by Grover search. Output $1$ if yes and $0$ if no.
\end{defaultenumerate}

Therefore, writing $a(n) \coloneqq \adv(f_n)$, observing that \Cref{eq:regular_recurrence} follows~\Cref{itm:strat1} and using \Cref{cor:strategies} gives $a(n)^2 \leq 2a(n/2)^2 + O(n)$, which solves to $a(n)^2 = O(n\log(n))$ by the Master Theorem (\cref{lemma:master_theorem}). Hence, $Q(f_n) = O(a(n)) = O(\sqrt{n\log(n)})$ by \Cref{thm:adv-q}.
\end{proof}

\subsection{String Rotation and String Suffix}
\label{subsec:stringrot}
Let $\Sigma$ be equipped with a total order in this subsection. We consider the following problems.

\begin{Problems*}[String Rotation and String Suffix] Let $n \in \mathbb{N}$ and $i\in [n]$. Given query access to a string $x\in \Sigma^n$, we define the following problems. 
\begin{enumerate}
    \item Minimal String Rotation. Decide if $x[i\twodots n]x[1 \twodots i-1] \leq x[j\twodots n]x[1 \twodots j-1]$ for all $j\in [n]$.
    \item Minimal Suffix. Decide if $x[i\twodots n] < x[j\twodots n]$ for all $j\in [n] \setminus \{i\}$.
\end{enumerate}
\end{Problems*} 

\begin{Remark*}
These problems are decision versions of the problems studied in \cite{string_akmaljin_2022}, which gives algorithms that output the positions of the minimal string rotation and minimal suffix.
\end{Remark*}

As observed in \cite{string_akmaljin_2022}, both problems reduce to the following problem.

\begin{Problem*}[Minimal Length-$l$ Substring] Let $n, l \in \mathbb{N}$ with $l\leq n$. Given query access to strings $x\in \Sigma^n$ and $y\in \Sigma^l$, decide if all length-$l$ substrings of $x$ are lexicographically at least $y$.
\end{Problem*}

We prove the following lemma by quantum divide and conquer.

\begin{Lemma}\label{lem:minimal_length_l}
The quantum query complexity of Minimal Length-$l$ Substring with $l=n/2$ is $\tilde{O}(\sqrt{n})$.
\end{Lemma}

\begin{proof}
We upper bound the quantum query complexity of the function $f_n \colon \Sigma^n\times \Sigma^{n/2} \to \{0,1\}$ defined by $f_n(x,y) = 1$ iff all length-$(n/2)$ substrings of $x$ are lexicographically at least $y$. 

Observe that $f_n(x,y) = 1$ if and only if (a) all length-$(n/4)$ substrings contained in $x[1 \twodots \sfrac{3n}{4}]$ are lexicographically at least $y[1 \twodots \sfrac{n}{4}]$; and (b) all length-$(n/2)$ substrings that start with $y[1 \twodots \sfrac{n}{4}]$ are lexicographically at least $y$. Also observe that all length-$(n/4)$ substrings contained in  $x[1 \twodots \sfrac{3n}{4}]$ are either contained in $x[1 \twodots \sfrac{n}{2}]$ or contained in $x[\sfrac{n}{4}+1 \twodots \sfrac{3n}{4}]$. Therefore, we deduce that
\begin{align}\label{eq:lminimal_recurrence}
    f_n(x,y) &= \Bigl(f_{n/2}(x[1\twodots \sfrac{n}{2}], y[1 \twodots \sfrac{n}{4}]) \wedge f_{n/2}(x[\sfrac{n}{4}+1 \twodots \sfrac{3n}{4}], y[1 \twodots \sfrac{n}{4}])\Bigr) \wedge g_n(x,y),
\end{align}
where $g_n\colon \Sigma^n \times \Sigma^{n/2} \to \{0,1\}$ is defined by
\begin{equation}
\begin{aligned}
    g_n(x,y) &=
    \begin{cases}
    1 &\begin{aligned}[t]
    &\text{if all length-$(n/2)$ substrings in $x$ starting with $y[1 \twodots \sfrac{n}{4}]$}
    \\
    &\text{are lexicographically at least $y$ (or do not exist),}
    \end{aligned}
    \\
    0 &\text{otherwise.}
    \end{cases}
\end{aligned}
\end{equation}

We can show that $Q(g_n) = \tilde{O}(\sqrt{n})$ by considering an algorithm $\mathcal{A}$ described as follows.
\begin{defaultenumerate}
    \item Use the quantum string matching algorithm (see \Cref{thm:algorithms}) to find the \emph{first} and \emph{last} positions where $y[1\twodots \sfrac{n}{4}]$ starts in $x[1\twodots 1 + \sfrac{n}{2}) = x[1\twodots \sfrac{n}{2}]$. Call these two positions $u_1$ and $v_1$, respectively. This takes $\tilde{O}(\sqrt{n})$ queries. Similarly, use the quantum string matching algorithm to find the first and last positions where $y[1\twodots \sfrac{n}{4}]$ starts in $x[\sfrac{n}{4}+1\twodots 1 + \sfrac{3n}{4}) = x[\sfrac{n}{4}+1\twodots \sfrac{3n}{4}]$. Call these two positions $u_2$ and $v_2$, respectively. This takes another $\tilde{O}(\sqrt{n})$ queries.
    \item Use the quantum string comparison algorithm (see \Cref{thm:algorithms}) to decide whether all the length-$(n/2)$ substrings in $x$ starting at one of the positions in $\{u_1,v_1,u_2,v_2\}$ is lexicographically at least $y$. If yes, then output $1$; otherwise, output $0$. This takes $\tilde{O}(\sqrt{n})$ queries.
\end{defaultenumerate}
Clearly, $\mathcal{A}$ uses $\tilde{O}(\sqrt{n})$ queries. The correctness of $\mathcal{A}$ follows from \cite[Lemma 4.8]{string_akmaljin_2022}.\footnote{Set $``l" = n/2$, $``k"=n/4$, and $``a" \in \{1,n/4\}$ in \cite[Lemma 4.8]{string_akmaljin_2022}. We also remark that \cite[Lemma 4.8]{string_akmaljin_2022} is stated with ``\ldots let $J$ denote the set of answers in the Minimal Length-$k$ Substrings problem on the input string $s[a\twodots a+2k)$.'' By examining the proof, it can be seen that the lemma still holds with the quoted sentence replaced by ``\ldots let $y$ denote a fixed length-$k$ substring, let $J$ denote the set of starting points of $y$ in the input string $s[a\twodots a+2k)$.'' This is because \cite[Lemma 2.3]{string_akmaljin_2022} can still be applied in the latter case.}

Now, setting $a(n) \coloneqq \adv(f_n)$, observing that $Q(g_n) = \tilde{O}(\sqrt{n})$, and using~\Cref{cor:strategies} for~\Cref{itm:strat1} gives $a(n)^2 = 2a(n/2)^2 + \tilde{O}(n)$, which solves to $a(n)^2 = \tilde{O}(n)$ by the Master Theorem (\Cref{lemma:master_theorem}). Hence, $Q(f_n) = O(a(n)) = \tilde{O}(\sqrt{n})$ by \Cref{thm:adv-q}.
\end{proof}

\begin{Remark*} The analysis in \cite[Theorem 4.1]{string_akmaljin_2022} for Minimal Length-$l$ Substring would yield a looser bound of $O(n^{1/2+o(1)})$ for its quantum query complexity. That analysis is more involved because it splits the problem into a non-constant number of parts, $m$. This is because quantum algorithms are recursively applied to obtain a recurrence of the form $Q(n) \leq \tilde{O}(\sqrt{m}) \, Q(n/m) + \tilde{O}(\sqrt{mn})$, which, due to the $\tilde{O}(\sqrt{m})$ coefficient in front of the $Q(n/m)$, requires $m$ to be non-constant to solve to $Q(n) = O(n^{1/2+o(1)})$. In general, it can be difficult to split a problem into a non-constant number of parts by only relying on quantum algorithmic techniques (e.g., consider an $\ANDOR$ formula), so such an approach may not yield optimal upper bounds (even up to an $n^{o(1)}$ factor). We note that our analysis above is a quantum analogue of the \emph{classical} divide-and-conquer analysis in \cite[Start of Section~4.2]{string_akmaljin_2022}.
\end{Remark*}

\Cref{thm:rotation_suffix} below follows from \Cref{lem:minimal_length_l} and the chain of reductions described in \cite[Propositions 4.2--4.5 and Theorem 4.6]{string_akmaljin_2022}. We sketch a proof for completeness.

\begin{Theorem}\label{thm:rotation_suffix}
The quantum query complexities of Minimal String Rotation and Minimal Suffix are $\tilde{O}(\sqrt{n})$.
\end{Theorem}

\begin{proof}[Proof sketch]
Let $f_{2n}$ be defined as in the proof of \Cref{lem:minimal_length_l}. The theorem follows from two observations:
\begin{enumerate}
    \item Minimal String Rotation can be computed by $f_{2n}(xx, x[i \twodots n]x[1 \twodots i-1])$.
    \item Minimal Suffix can be computed by $f_{2n}(1x0^{n-1}, x[i \twodots n]0^{i-1})$, where $0$ is defined to be smaller than all elements in $\Sigma$ and $1$ is defined to be larger than all elements in $\Sigma$.\qedhere
\end{enumerate}
\end{proof}

\begin{Remark*}
The quantum query complexities of Maximal String Rotation and Maximal Suffix (defined in the obvious way) are also $\tilde{O}(\sqrt{n})$ because they  easily reduce to Minimal String Rotation and Minimal Suffix.
\end{Remark*}

\subsection{\texorpdfstring{$k$}{k}-Increasing Subsequence}
\label{subsec:kis}

Let $\Sigma$ be equipped with a total order and fix $k$ to be a constant in this subsection. We consider the following problem, which is a natural parametrized version of Longest Increasing Subsequence.

\begin{Problem*}[$k$-IS] Let $n\in \mathbb{N}$. Given query access to $x\in \Sigma^n$, decide if $x$ has a $k$-IS, i.e., if there exist integers $i_1 < i_2 < \cdots < i_k$ such that $x[i_1] < x[i_2] < \cdots < x[i_k]$.
\end{Problem*}

To solve $k$-IS, we solve a variant of it that we call $k$-Increasing Subsequence* ($k$-IS*), defined as follows. We consider  $k$-IS* because it is more susceptible to recursion.

\begin{Problem*}[$k$-IS*] Let $n\in \mathbb{N}$ and let $*$ denote an element outside $\Sigma$. Given query access to $x\in (\Sigma\cup \{*\})^n$, decide if $x$ has a $k$-IS*, i.e., if there exist integers $i_1 < i_2 < \cdots <i_k$ such that $x[i_1] < x[i_2] < \cdots < x[i_k]$ and $x[i_l]\neq *$ for all $l\in [k]$.
\end{Problem*}

We now prove the main theorem of this subsection.

\begin{Theorem}
The quantum query complexities of \textup{$k$-IS} and \textup{$k$-IS*} are both $O(\sqrt{n}\log^{3(k-1)/2}(n))$.
\end{Theorem}

\begin{proof}
Clearly, $k$-IS reduces to $k$-IS* without using additional queries. Therefore, it suffices to prove the theorem for $k$-IS*. We upper bound the quantum query complexity of the function $\LIS_{k,n}\colon \Sigma^n \to \{0,1\}$ defined by $\LIS_{k,n}(x) = 1$ iff $x$ contains a $k$-IS*.

We first introduce some definitions. Define $a_k(n)\coloneqq \adv(\LIS_{k,n})$. Define $\LIS_{(i,j),n}\colon \Sigma^n \to \{0,1\}$ by $\LIS_{(i,j),n}(x) = 1$ if and only if $x$ contains an $(i+j)$-IS* consisting of an $i$-IS*, $u_1 < u_2 < \cdots < u_i$, and a $j$-IS*,  $v_1 < v_2 < \cdots < v_j$, such that $u_i \leq n/2 < v_1$. For $j\in \mathbb{N}$, define $\mt_{j,n}\colon \Sigma^n \to \Sigma$ by $\mt_{j,n}(x) = \min_l x[l]$, where the minimization is over those $l\in [n]$ that are the last index of some $j$-IS* in $x$. Similarly, define $\mh_{j,n}\colon \Sigma^n \to \Sigma$ by $\mh_{j,n}(x) = \max_l x[l]$, where the maximization is over those $l\in [n]$ that are the first index of some $j$-IS* in $x$.

Observe that the increasing subsequence can be (i) contained entirely in the left half of $x$; or (ii) contained entirely in the right half of $x$; or (iii) partially contained in both halves with $i$ elements in the left and $k-i$ elements in the right for some $i \in [k-1]$. Therefore, we have
\begin{equation}
    \LIS_{k,n}(x) =  \LIS_{k,n/2}(x_\textleft) \vee \LIS_{k,n/2}(x_\textright) \vee \bigvee_{i=1}^{k-1} \LIS_{(i,k-i),n}(x).
\end{equation}

Therefore, by \Cref{cor:strategies} for~\Cref{itm:strat1}, we obtain the quantum divide and conquer recurrence
\begin{equation}\label{eq:lis-recurrence}
    a_k(n)^2 \leq 2 a_k(n/2)^2 + \sum_{i=1}^{k-1} O(Q(\LIS_{(i,k-i),n})^2).
\end{equation}

Now,
\begin{equation}\label{eq:lis-combine}
    Q(\LIS_{(i,k-i),n}) \leq Q(\mt_{i,n/2}) +  Q(\mh_{k-i,n/2}),
\end{equation}
because $\LIS_{(i,k-i),n}(x)$ can be computed by computing $\mt_{i,n/2}(x_\textleft)$ and $\mh_{k-i,n/2}(x_\textright)$ and checking that $\mt_{i,n/2}(x_\textleft)<\mh_{k-i,n/2}(x_\textright)$.

Moreover, for any $j\in \mathbb{N}$, $\mt_{j,n}$ and $\mh_{j,n}$ can be computed by a ``randomized search'' that uses $O(\log(n))$ computations of $\LIS_{j,n}$ and Grover search, which is similar to the approach taken for quantum minimum finding~\cite{minimum_finding_durr_hoyer_1996}---see~\Cref{app:randomized-reduction} for a proof. Therefore, for all $j,n\in \mathbb{N}$, we can use \cref{thm:adv-q} to obtain
\begin{equation}\label{eq:lis-search}
    Q(\mt_{j,n}), Q(\mh_{j,n}) \leq O((a_j(n) + \sqrt{n}) \log(n)).
\end{equation}

Substituting the combination of \Cref{eq:lis-combine} and \Cref{eq:lis-search} into \Cref{eq:lis-recurrence} shows that, for $k\geq 2$, we have
\begin{equation}\label{eq:lis-recurrence-2}
    a_k(n)^2 \leq 2 a_k(n/2)^2 + O( a_{k-1}(n)^2\log^2(n)),
\end{equation}
where we used the fact that $k$ is constant and $\Omega(\sqrt{n}) \leq a_i(n) \leq O(a_{i+1}(n))$ for all $i\geq 1$, which follows by applying \Cref{thm:adv-q} to the observations that (i) $\LIS_{i,n}(x)= \LIS_{i+1,n+1}(0x)$ for all $x$, where $0\neq *$ is some element smaller than all elements in $\Sigma$, and (ii) computing $\LIS_{1,n}$ is equivalent to searching for a $*$.

Finally, we prove $a_k(n) = O(\sqrt{n}\log^{3(k-1)/2}(n))$ by induction on $k\geq 1$. The base case ($k=1$) follows by applying \Cref{thm:adv-q} to the observation that computing $\LIS_{1,n}$ is equivalent to searching for a $*$. For $k>1$, \Cref{eq:lis-recurrence-2} and the inductive hypothesis for the $(k-1)$-th case give
\begin{equation}\label{eq:lis-recurrence-3}
    a_k(n)^2 \leq 2 a_k(n/2)^2 + O(n\log^{3k-4}(n)),
\end{equation}
which solves to $a_k(n) = O(\sqrt{n}\log^{3(k-1)/2}(n))$, as required. The result follows from \Cref{thm:adv-q}.
\end{proof}

\subsection{\texorpdfstring{$k$}{k}-Common Subsequence}
\label{subsec:kcs}

As our last application, we consider a parameterized version of the Longest Common Subsequence (LCS) problem. We fix $k$ to be a constant in this subsection.

\begin{Problem*}[$k$-CS] Let $n\in \mathbb{N}$. Given query access to strings $x, y\in \Sigma^n$, decide if they share a $k$-common subsequence ($k$-CS), i.e., if there exist integers $i_1<i_2< \cdots <i_k$ and $j_1<j_2< \cdots <j_k$ with $x[i_l]=y[j_l]$ for all $l \in [k]$.
\end{Problem*}

The main theorem of this subsection is the following.

\begin{restatable}{Theorem}{theoremKCSMainClaim}\label{thm:kcs_main_claim}
The quantum query complexity of $k$-CS is $O(n^{2/3} \log^{k-1}(n))$.
\end{restatable}

Define $\LCS_{k,n}\colon \Sigma^n\times \Sigma^n \to \{0,1\}$ by $\LCS_{k,n}(x,y) = 1$ iff $x, y \in \Sigma^n$ share a $k$-CS. When $k = 1$ this problem is equivalent to the (bipartite) Element Distinctness problem \cite{collision_aaronson_shi_2004, qwalk_elem_distictness_ambainis_2007}. Aaronson and Shi showed that the quantum query complexity of Element Distinctness is at least $\Omega(n^{2/3})$ \cite{collision_aaronson_shi_2004}, while Ambainis proved the matching upper bound of $O(n^{2/3})$ \cite{qwalk_elem_distictness_ambainis_2007}. Since Element Distinctness reduces to $k$-CS (by prefixing $k-1$ common elements to the input strings of $k$-CS), we see that our theorem gives a tight bound on the quantum query complexity of $k$-CS up to logarithmic factors.

We start by introducing some terminology. For $x,y \in \Sigma^n$, we say that $(i,j) \in [n] \times [n]$ is a collision of $x$ and $y$ iff $x[i] = y[j]$ (see \Cref{fig:collision_types} for examples). 

\begin{figure}[ht!]
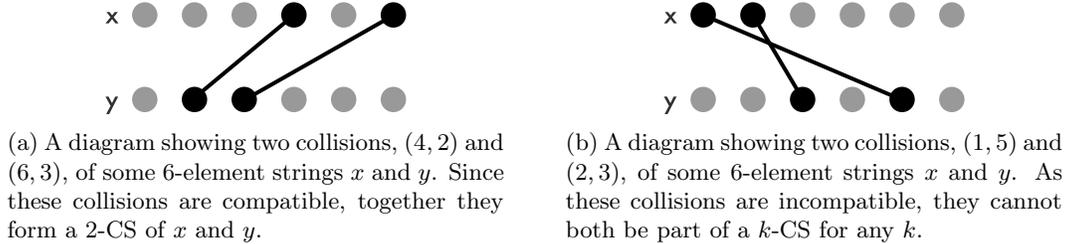

\centering
\begin{subfigure}[t]{0.4\textwidth}
    \centering
    \includesvg[scale=0.5]{figures/compatible_collisions.svg}
    \caption{A diagram showing two collisions, $(4,2)$ and $(6,3)$, of some $6$-element strings $x$ and $y$. Since these collisions are compatible, together they form a $2$-CS of $x$ and $y$.}
    \label{subfig:compatible_collisions_a}
\end{subfigure}
\qquad
\begin{subfigure}[t]{0.4\textwidth}
    \centering
    \includesvg[scale=0.5]{figures/crossing_collisions.svg}
    \caption{A diagram showing two collisions, $(1,5)$ and $(2,3)$, of some $6$-element strings $x$ and $y$. As these collisions are incompatible, they cannot both be part of a $k$-CS for any $k$.}
    \label{subfig:compatible_collisions_b}
\end{subfigure}
\caption{Diagrams illustrating collisions of $6$-element strings.}
\label{fig:collision_types}
\end{figure} 

We fix some constant $m \in \mathbb{N}$ which we call the splitting factor. Given inputs $x,y\in\Sigma^n$, we consider splitting $x$ and $y$ each into $m$ substrings of size $n/m$. We denote these substrings by
\begin{align}
    x^{(i)} \coloneqq x\left(\frac{i-1}{m}n\twodots \frac{i}{m}n\right] \quad \text{and} \quad y^{(i)} \coloneqq x\left(\frac{i-1}{m}n\twodots \frac{i}{m}n\right] \quad \text{for all $i \in [m]$}.
\end{align}

\begin{Definition}[Subproblems and signatures]
By a \emph{subproblem}, we refer to a tuple $(i,j)\in [m]\times[m]$. By the \emph{$(i,j)$-subproblem of $(x,y)\in \Sigma^n \times \Sigma^n$}, we refer to the tuple $(x^{(i)}, y^{(j)})$. To each $x,y\in \Sigma^n$ we associate an $m^2$-bit signature $\mathbf{s} \in \mathcal{S} \coloneqq \{0,1\}^{m^2}$ such that $\s_{i,j} = 1$ iff there is a collision between $x^{(i)}$ and $y^{(j)}$. Let $\sigma_n \colon \Sigma^n \times \Sigma^n \to \mathcal{S}$ be the function that, given input $(x,y)$, returns its signature.
\end{Definition}

It is clear that the signature can be computed by using Ambainis's algorithm for Element Distinctness~\cite{qwalk_elem_distictness_ambainis_2007} $m^2$ times, once for each subproblem. Since $m$ is a constant, we have the following.

\begin{Lemma}\label{lemma:adv_signature}
The quantum query complexity of $\sigma_n$ satisfies $Q(\sigma_n) = O(n^{2/3})$.
\end{Lemma}

\begin{Definition}[Compatible, relevant, and critical subproblems]\
\begin{enumerate}[(i),topsep=4pt]
    \item We say that two subproblems $(i_1, j_1)$ and $(i_2, j_2)$ are \emph{compatible} if ($i_1 < i_2$ and $j_1 < j_2$) or ($i_1 > i_2$ and $j_1 > j_2$). 
    \item Given a signature $\s\in \mathcal{S}$, we say that the $(i,j)$ subproblem is \emph{$\s$-relevant} if $\s_{i,j} = 1$.
    \item Given a signature $\s\in \mathcal{S}$, we say that an $\s$-relevant subproblem $(i,j)$ is \emph{$\s$-critical} if none of the subproblems compatible with $(i,j)$ are $\s$-relevant.  
\end{enumerate}
\end{Definition}

Observe that if subproblems $(i_1, j_1)$ and $(i_2, j_2)$ are compatible then, for any $x, y\in \Sigma^n$, the union of any $k_1$-CS of the $(i_1,j_1)$-subproblem of $(x,y)$ and any $k_2$-CS of the $(i_2,j_2)$-subproblem of $(x,y)$ is a $(k_1 + k_2)$-CS of $(x,y)$.
A key intuition of our approach is the following.
Given $x,y\in \Sigma^n$, to decide whether $x,y$ share a $k$-CS using divide and conquer, we might want to consider recursing on the $(i_1, j_1)$-subproblem of $(x,y)$. Assume that we have obtained the signature $\s$ of $(x,y)$, and we find that  subproblem $(i_1, j_1)$ is $\s$-relevant but not $\s$-critical, i.e., there exists an $\s$-relevant subproblem $(i_2, j_2)$ which is compatible with $(i_1, j_1)$. In this setting, finding even a $(k-1)$-CS of the $(i_1, j_1)$-subproblem of $(x,y)$ would suffice to conclude that $\LCS_{k,n}(x,y) = 1$, since the collision in the $(i_2, j_2)$-subproblem of $(x,y)$ can be combined with a $(k-1)$-CS in the $(i_1, j_1)$-subproblem of $(x,y)$ to give a $k$-CS of $(x,y)$. This suggests that we should only recurse on those subproblems that are critical according to the signature of $(x,y)$. 

The above discussion motivates the following definition.

\begin{Definition}[Simple and composite $k$-CS]\label{def:simple_composite}
Let $x, y \in \Sigma^n$. We say that a $k$-CS of $(x, y)$ is \emph{simple} if it is also a $k$-CS of $(x^{(i)},y^{(j)})$ for some $i,j \in [m]$. A $k$-CS of $(x,y)$ which is not simple is called \emph{composite}.
\end{Definition}

\begin{figure}[ht!]
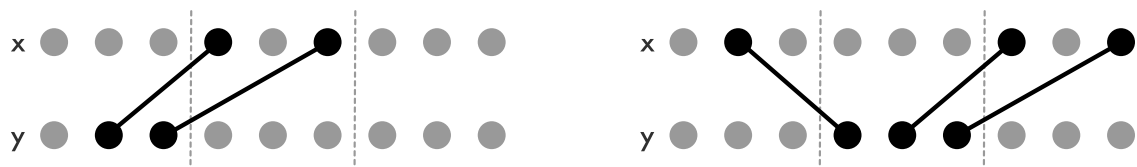

\centering
\begin{subfigure}[t]{0.45\textwidth}
    \includesvg[scale=0.55]{figures/simple_witness.svg}
    \caption{A diagram illustrating a $2$-CS of some $9$-element strings $x$ and $y$. When the splitting factor is $3$, this $2$-CS is a simple witness of $\LCS_{2,9}(x,y) = 1$, since it is also a $2$-CS of the $(2,1)$-subproblem of $(x,y)$.} 
    \label{subfig:simple_witness}
\end{subfigure}
\hspace{2em}
\begin{subfigure}[t]{0.45\textwidth}
    \includesvg[scale=0.55]{figures/composite_witness.svg}
    \caption{A diagram illustrating a $3$-CS of some $9$-element strings $x$ and $y$. When the splitting factor is $3$, this $3$-CS is a composite witness of $\LCS_{3,9}(x,y) = 1$.}
    \label{subfig:composite_witness}
\end{subfigure}
\caption{Diagrams illustrating simple and composite witnesses for different choices of $k$ with respect to a splitting factor of $m = 3$.}
\label{fig:simple_and_composite_witnesses}
\end{figure}

We define $\LCSC{k}{n}\colon\Sigma^n\times \Sigma^n \to \{0,1\}$ by $\LCSC{k}{n}(x,y) = 1$ iff $x$ and $y$ share a composite $k$-CS (see \Cref{fig:simple_and_composite_witnesses}). Given a signature $\s \in \mathcal{S}$, define $h_{k,n}^{(\s)} \colon \Sigma^n \times \Sigma^n \to \{0,1\}$ by $h_{k,n}^{(\s)}(x,y) = 1$ iff there exists an $\s$-critical subproblem $(i, j)$ such that the $(i,j)$-subproblem of $(x,y)$ has a $k$-CS.

By the above discussion, we have the following proposition.

\begin{Proposition}\label{lemma:kcs_case_switch}
For all $x,y\in \Sigma^n$, we have
\begin{align}\label{eq:kcs_decomposition}
   \LCS_{k,n}(x,y) = \LCSC{k}{n}(x,y) \lor h_{k,n}^{(\sigma_n(x,y))}(x,y).
\end{align}
\end{Proposition}
\begin{proof}
If $x$ and $y$ do not share a $k$-CS, then both sides of the equation evaluate to false. If $x$ and $y$ share a composite $k$-CS then both sides of the equation evaluate to true. 

Consider the remaining case when $x$ and $y$ share a simple $k$-CS but no composite $k$-CS. The left-hand side of \Cref{eq:kcs_decomposition} evaluates to true.  Let $\s \coloneqq \sigma_n(x,y)$ be the signature of $(x,y)$. Then, the right-hand side of \Cref{eq:kcs_decomposition} evaluates to $h_{k,n}^{(\s)}(x,y)$. Since $x$ and $y$ share a simple $k$-CS, there exists a $k$-CS of $(x,y)$ which is also a $k$-CS of $(x^{(i)}, y^{(j)})$ for some $\s$-relevant subproblem $(i, j)$. Then $(i,j)$ must also be $\s$-critical, since otherwise there would be an $\s$-relevant subproblem $(i', j')$ compatible with $(i,j)$, which implies that the $(i,j)$- and $(i', j')$-subproblems of $(x,y)$ would together yield a composite $k$-CS of $(x,y)$, contradicting the case we are in. Therefore, $h_{k,n}^{(\mathbf{s})}(x,y)$ also evaluates to true as required.
\end{proof}

Let $a_k(n)$ denote the adversary quantity of $\LCS_{k,n}$. As $Q(\LCS_{k,n}) = \Theta(a_k(n))$ (\cref{thm:adv-q}), we can prove \Cref{thm:kcs_main_claim} by  upper bounding the adversary quantity $a_k(n)$ instead of the quantum query complexity $Q(\LCS_{k,n})$. As a first step we establish upper bounds on the adversary quantities and quantum query complexities of the functions on the right-hand side of \Cref{eq:kcs_decomposition}.

\begin{Lemma}\label{lem:upper_bound_adv_hkns}
For any signature $\mathbf{s}\in \mathcal{S}$, the adversary quantity of the function $h_{k,n}^{(\s)}$ satisfies
\begin{align}
    \adv(h_{k,n}^{(\mathbf{s})}) \leq \sqrt{2m-1} \, a_k(n/m).
\end{align}
\end{Lemma}

\begin{Lemma}\label{lem:upper_bound_adv_fknc} The quantum query complexity of $\LCSC{k}{n}$ satisfies
\begin{align}
    Q(\LCSC{k}{n}) \leq O(a_{k-1}(n)\log{n}).
\end{align}
\end{Lemma}

We defer the proofs of \Cref{lem:upper_bound_adv_hkns,lem:upper_bound_adv_fknc} for now, and prove the following lemma assuming these results. 

\begin{Lemma}\label{lemma:recurrence_of_kcs_adv}
For $k\geq 2$, the adversary quantity $a_{k}(n)$ of $\LCS_{k,n}$ satisfies the recurrence
\begin{align}
    a_k(n) \leq \sqrt{2m-1}\ a_k(n/m) + O(a_{k-1}(n) \log{n}).
\end{align}
\end{Lemma}
\begin{proof}
Recall that by \Cref{lemma:kcs_case_switch},
\begin{align}
\LCS_{k,n}(x,y) = \LCSC{k}{n}(x,y) \lor h_{k,n}^{(\sigma_n(x,y))}(x,y).
\end{align}
This Boolean expression is an $\OR$ where one of the components, $h_{k,n}$, is a $\SWITCH$ conditioned on the value of the signature function $\sigma_n(x,y)$. This involves both of our strategies discussed in~\Cref{sec:framework}. Therefore, by applying \Cref{cor:strategies},
\begin{align}
    a_k(n) = \adv(\LCS_{k,n}) &\leq O(Q(\LCSC{k}{n})) +
    O(Q(\sigma_n)) +
    \max_{\mathbf{s}\in\mathcal{S}}\{\adv(h_{k,n}^{(\mathbf{s})})\}. \\ 
    \intertext{By \Cref{lemma:adv_signature,lem:upper_bound_adv_hkns,lem:upper_bound_adv_fknc},} 
    a_k(n) &\leq O(a_{k-1}(n)\log{n}) +  O(n^{2/3}) + \sqrt{2m-1} \, a_{k}(n/m) \\
    &\leq \sqrt{2m-1} \, a_{k}(n/m) + O(a_{k-1}(n)\log{n}),
\end{align}
where the last line follows since $a_1(n) \in \Theta(n^{2/3})$, and $a_{k-1}(n) = \Omega(n^{2/3})$ since Element Distinctness reduces to $(k-1)$-CS for $k\geq 2$ by prefixing $k-2$ common elements to the input strings for $(k-1)$-CS. 
\end{proof}

The above recurrence allows us to prove our main theorem.
\begin{proof}[Proof of \Cref{thm:kcs_main_claim}]
We proceed by induction on $k\geq 1$. In the base case ($k=1$), the theorem holds since $1$-CS is equivalent to bipartite Element Distinctness, which has quantum query complexity $O(n^{2/3})$ \cite{qwalk_elem_distictness_ambainis_2007}. Assume that $k\geq 2$. By \cref{lemma:recurrence_of_kcs_adv} we have the following recurrence for the adversary quantity of $\LCS_{k,n}$:
\begin{align}
    a_k(n) &\leq \sqrt{2m-1} \, a_{k}(n/m) + O(a_{k-1}(n)\log{n}) \\ 
    &\leq \sqrt{2m-1} \, a_k(n/m) + O(n^{2/3}\log^{k-1}(n))
\end{align}
where the second inequality follows by our inductive hypothesis, $a_{k-1}(n) \leq O(n^{2/3} \log^{k-2}(n))$. The result follows from the Master Theorem (\Cref{lemma:master_theorem}) provided that $\log_m(\sqrt{2m-1}) < 2/3$. We achieve this by choosing $m = 7$.
\end{proof}

\begin{figure}[ht!]
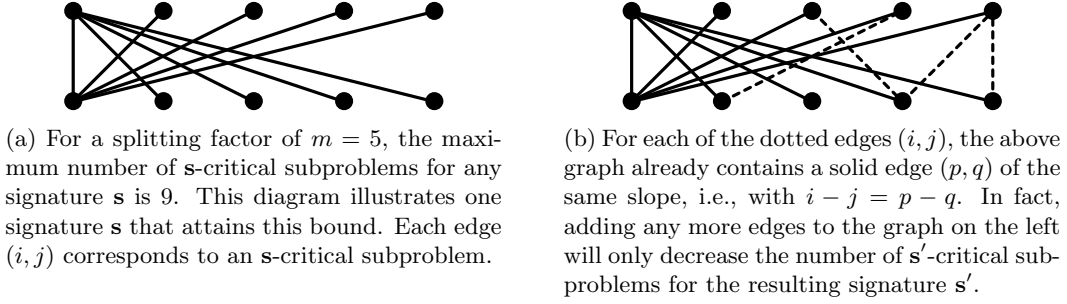

\centering
\begin{subfigure}[t]{0.4\textwidth}
    \centering
    \includesvg[scale=1.6]{figures/graph_with_maximal_number_of_critical_edges.svg}
    \caption{For a splitting factor of $m = 5$, the maximum number of $\s$-critical subproblems for any signature $\s$ is 9. This diagram illustrates one signature $\s$ that attains this bound. Each edge $(i,j)$ corresponds to an $\s$-critical subproblem.}
    \label{subfig:assoc_graph_extremal_example}
\end{subfigure}
\qquad
\begin{subfigure}[t]{0.4\textwidth}
    \centering
    \includesvg[scale=1.6]{figures/graph_with_maximal_number_of_critical_edges2.svg}
    \caption{For each of the dotted edges $(i,j)$, the above graph already contains a solid edge $(p,q)$ of the same slope, i.e., with $i-j = p-q$. In fact, adding any more edges to the graph on the left will only decrease the number of $\s'$-critical subproblems for the resulting signature $\s'$.}
    \label{subfig:assoc_graph_extremal_example2}
\end{subfigure}
\caption{Diagrams illustrating an example of a signature $\s$ for which the number of $\s$-critical subproblems is maximal. Here, the splitting factor is $m = 5$ and we represent a signature $\s\in \mathcal{S} = \{0,1\}^{m^2}$ by the subgraph of $K_{m,m}$ that has edge $(i,j)$ iff $\s_{i,j} = 1$.}
\label{fig:extremal}
\end{figure}

We now prove \Cref{lem:upper_bound_adv_hkns,lem:upper_bound_adv_fknc}.

\begin{proof}[Proof of \Cref{lem:upper_bound_adv_hkns}.] Fix an arbitrary signature $\s\in \mathcal{S}$. Recall that the function $h_{k,n}^{(\s)}(x,y)$ indicates whether there exists an $\s$-critical subproblem $(i,j)$ such that the $(i,j)$-subproblem of $(x,y)$ has a $k$-CS.
Let $R \coloneqq \{(i_l, j_l): l\in[r]\}$ denote the set of $\s$-critical subproblems. Then
\begin{align}
    h_{k,n}^{(\s)}(x,y) = \bigvee_{l\in [r]} \LCS_{k, n/m}(x^{(i_l)}, y^{(j_l)}).
\end{align}
By the adversary composition lemma for $\OR$ (\Cref{lem:and-or}),
\begin{align}
    \adv(h_{k,n}^{(\s)}) = \sqrt{r} \, a_k(n/m).
\end{align}

To conclude, it suffices to prove that $r \leq 2m-1$. To this end, to each subproblem $(i,j)$ we associate the (signed) slope $i-j$. Observe that there are only $2m-1$ possible slope values, namely $\{0, \pm1, \ldots, \pm (m-1)\}$. Now notice that if two distinct subproblems $(i,j)$ and $(i',j')$ have the same associated slope (i.e., $i-j = i'-j'$), then these subproblems are compatible (see \Cref{fig:extremal}). Therefore, for each slope, there can be at most one subproblem in $R$ since the subproblems in $R$ are $\s$-critical. Therefore, $r \leq 2m-1$.
\end{proof}

\begin{proof}[Proof of \Cref{lem:upper_bound_adv_fknc}.]
Label the vertices in the two parts of the complete bipartite graph $K_{m,m}$ by $u_i$ and $v_i$, respectively, for $i\in[m]$.

To any $k$-common subsequence $S$ of some $x$ and $y$, we assign a subgraph $G_S$ of $K_{m,m}$ that is obtained by removing all edges $(i,j)$ of $K_{m,m}$ for which $S$ does not contain a collision of $x^{(i)}$ and $y^{(j)}$. Then the graph $G_S$ has at least 2 edges if and only if $S$ is a composite $k$-CS; otherwise $G_S$ has a single edge only. We assign positive integer weights to the edges of $G_S$ by letting the weight of the edge $(u_i,v_j)$ be the number of collisions between the substrings $x^{(i)}$ and $y^{(j)}$ that are contained in $S$. Therefore, the total weight of $G_S$ is $k$ for any $k$-CS considered.

\begin{figure}[ht!]
\centering
\begin{subfigure}{0.9\textwidth}
    \centering
    \includesvg[scale=0.55]{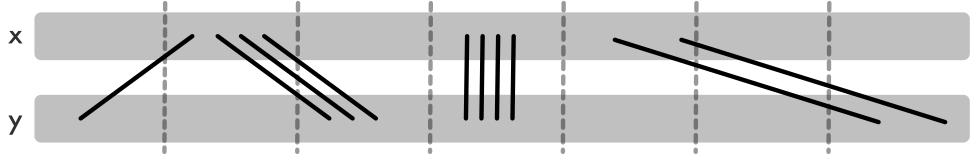}
     \caption{A schematic diagram illustrating a $10$-CS of some strings $x$ and $y$. The dotted lines break up the strings $x$ and $y$ into $m=7$ substrings of equal size.}
    \label{subfig:schematic_diagram_of_kcs}
\end{subfigure}\\
\vspace{1em}
\begin{subfigure}{0.9\textwidth}
    \centering
    \includesvg[scale=1.6]{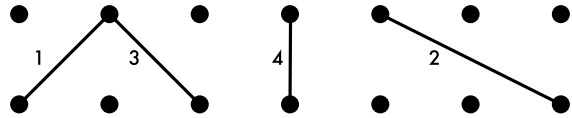}
    \caption{A diagram showing the graph $G_S$ associated with the $10$-CS seen in the above diagram, assuming a splitting factor of $m=7$.}
    \label{subfig:graph_rep_of_kcs}
\end{subfigure}
\caption{Diagrams illustrating a $10$-CS $S$ of some pair of strings, as well as the associated subgraph $G_S$ of the complete bipartite graph $K_{m,m}$. Here, the splitting factor is $m = 7$.}
\label{fig:kcs_graph_representation}
\end{figure}

Placing the vertices of $G_S$ on a square grid (say, $u_i$ on $(i,1)$ and $v_j$ on $(j,0)$ of the Euclidean plane) and drawing the edges as the line segments between the corresponding vertices, we see that by the definition of a $k$-CS, the edges of $G_S$ meet only at vertices in this drawing (i.e., we have a planar embedding, see \Cref{subfig:graph_rep_of_kcs}). This means that we can make the edge set of $G_S$ a totally ordered set, by letting $(u_{i},v_{j}) \preceq (u_{i'},v_{j'})$ iff $i \leq i'$ and $j \leq j'$. Therefore, a unique minimal (i.e., leftmost) edge exists. Let $(u_{i^{*}},v_{j^{*}})$ be this edge. We claim that at least one of $u_{i^*}$ and $v_{j^*}$ has degree 1. To see this, we can assume without loss of generality that $i^* \leq j^*$, as an analogous argument works for the other case. Suppose now that $u_{i^*}$ has a neighbor other than $v_{j^*}$. Denoting this other neighbor by $v_{j'}$, the minimality of the edge $(u_{i^*},v_{j^*})$ implies $j^{*} < j'$. Consider the vertex $v_{j^*}$, and suppose that $v_{j^*}$ has another neighbor $u_{i'}$ for some $i'$ with $i^* < i'$. Observe that in this case the edges $(u_{i^*}, v_{j'})$ and $(u_{i'}, v_{j^*})$ would be crossing in our drawing, contradicting our observation that there exists a planar embedding of $G_S$ (see \Cref{subfig:graph_rep_of_kcs} for an illustration). 

To summarize, to any $k$-CS $S$ we associate an edge-weighted undirected subgraph $G_S$ of $K_{m,m}$ that satisfies the following properties:
\begin{enumerate}[(i)]
    \item \label[property]{itm:weight} The weight of the edge $(u_i,v_j)$ in $G_S$ is the number of collisions of $x^{(i)}$ and $y^{(j)}$ that $S$ contains.
    \item \label[property]{itm:2edges} The graph $G_S$ has at least 2 edges if and only if $S$ is a composite $k$-CS.
    \item \label[property]{itm:uniqueleft} There is a unique leftmost edge $(u_{i^*},v_{j^*})$ of $G_S$. 
    \item \label[property]{itm:degreeone} The leftmost edge of $G_S$ is such that at least one of the vertices it is incident on has degree 1. 
\end{enumerate}

Let $f^{(i, j)}_{k,n}(x,y)$ denote the Boolean function that indicates if $x$ and $y$ have a composite $k$-CS for which $(i,j)$ is the leftmost edge of the associated graph $G_S$. By \Cref{itm:uniqueleft}, 
\begin{align}
    \LCSC{k}{n}(x,y) = \bigvee_{i,j\in[m]} f^{(i,j)}_{k,n}(x,y).
\end{align}
Therefore, an algorithm that decides $f^{(i,j)}_{k,n}(x,y)$ for every choice of $i,j \in [m]$ can decide $\LCSC{k}{n}(x,y)$. To conclude, we give a quantum algorithm that decides $f_{k,n}^{(i,j)}$ using $O(a_{k-1}(n) \log{n})$ queries, for any $i, j \in [m]$. As there are only $O(1)$ choices of $i,j\in[m]$ to consider (since $m$ is constant), an algorithm that iterates over each of these cases has the same asymptotic complexity.

Let $S$ be a witness of $f_{k,n}^{(i,j)}(x,y) = 1$ for some $i, j \in [m]$ and consider the associated graph $G_S$.
By \Cref{itm:degreeone} above, at least one of the vertices $u_{i}$ and $v_{j}$ has degree $1$ in $G_S$. Suppose that $v_j$ has degree $1$ (the other case can be treated analogously). Then by the properties of $G_S$ we observed earlier, we see that $S$ has the following structure (see \Cref{subfig:schematic_diagram_of_kcs} for an example):
\begin{enumerate}
    \item By \Cref{itm:2edges}, $G_S$ has at least 2 edges. Therefore by \Cref{itm:weight}, the weight  of $(u_i, v_j)$ is some $k_1 < k$.  
    Therefore, $S$ contains $k_1 \in [k-1]$ collisions of $x^{(i)}[1\twodots p]$ and $y^{(j)}$, for some index $p \in [n/m]$. 
    \item By \Cref{itm:degreeone}, and our assumption that $v_{j}$ has degree 1, $S$ does not contain any collisions between $y^{(j)}$ and any substring $x^{(i')}$ other than $x^{(i)}$. Therefore, none of the remaining $k-k_1$ collisions of $S$ involve $y^{(j)}$. In other words, the remaining $k-k_1$ collisions of $S$ are collisions of $x((\frac{i-1}{m}n + p)\twodots n]$ and $y(\frac{j}{m}n\twodots n]$.
\end{enumerate}
The following quantum query algorithm decides if a witness $S$ of the above structure exists, and returns false otherwise, using $O(a_{k-1}(n) \log{n})$ queries:
\begin{enumerate}
    \item By combining a binary search over the indices of $x^{(i)}$ with a query-optimal quantum algorithm for $k_1$-CS, identify the minimal $p \in [n/m]$ for which $x^{(i)}[1\twodots p]$ and $y^{(j)}$ have a $k_1$-CS. If no such $p$ has been found, return false. This can be done using $O(a_{k_1}(n) \log{n})$ queries.\footnote{Since the quantum algorithm for $k_1$-CS can be erroneous with bounded probability, one might naively expect an additional $\log\log(n)$ factor due to the need for error reduction since the algorithm for $k_1$-CS is called $O(\log(n))$ times in the binary search. However, \cite{noisybinary_feige_1994} shows that the $\log\log(n)$ factor can be removed by performing a variant of binary search.}
    \item Using a query-optimal algorithm for $(k-k_1)$-CS, decide if $x((\frac{i-1}{m}n+p)\twodots n]$ and $y(\frac{j}{m}n\twodots n]$ have a $(k-k_1)$-CS, and return the output. This can be done using $O(a_{k-k_1}(n))$ queries.
\end{enumerate}
As $1 \leq k_1\leq k-1$, and $a_l(\cdot)$ is a non-decreasing\footnote{Recall that $(l-1)$-CS reduces to $l$-CS by prefixing a common element to the input strings of $l$-CS.} function of $l$, the total query complexity of the above algorithm is
\begin{align}
    \oh{a_{k_1}(n)\log{n} + a_{k-k_1}(n)} &\leq \oh{a_{k-1}(n) \log{n}}.
\end{align}

To decide $f_{k,n}^{(i,j)}$, we run the above algorithm for every possible choice of $k_1 \in [k-1]$. Since we arbitrarily assumed it was $v_j$ that had degree 1, we run an analogous algorithm assuming that $u_i$ has degree 1. Overall this does not affect the asymptotic query complexity of deciding $f_{k,n}^{(i,j)}$, which is
\begin{equation}
    Q(f_{k,n}^{(i,j)}) \leq 2(k-1) \cdot O(a_{k-1}(n) \log{n}) \leq O(a_{k-1}(n) \log{n}).
\end{equation}
The lemma follows.
\end{proof}


\section{Open problems}
\label{sec:future}

We conclude by describing some open problems related to possible extensions of this work.

\begin{itemize}
    \item Can we apply quantum divide and conquer to search problems? For example, is there a quantum divide-and-conquer algorithm for minimum finding?
    \item Can we find applications of quantum divide and conquer using generalizations of the two strategies presented in \Cref{sec:framework} (i.e., using combining functions other than $\ANDOR$ formulas and $\SWITCH$, such as those discussed at the end of \Cref{sec:techniques})?
    \item Can we obtain super-quadratic speedups using quantum divide and conquer?
\end{itemize}


\section*{Acknowledgements}

We thank Saeed Seddighin for telling us about the problem of determining the quantum query complexity of $k$-Increasing Subsequence, which became the starting point for this work.

We acknowledge support from the Army Research Office (grant W911NF-20-1-0015); the Department of Energy, Office of Science, Office of Advanced Scientific Computing Research, Accelerated Research in Quantum Computing program; and the National Science Foundation (grants CCF-1813814 and DMR-1747426).


\bibliography{references}

\newcommand{\etalchar}[1]{$^{#1}$}
\newcommand{\stoc}[1]{Proceedings of the #1 {ACM} Symposium on Theory of
  Computing ({STOC})}\newcommand{\focs}[1]{Proceedings of the #1 {IEEE}
  Symposium on Foundations of Computer Science
  ({FOCS})}\newcommand{\soda}[1]{Proceedings of the #1 {ACM-SIAM} Symposium on
  Discrete Algorithms ({SODA})}\newcommand{\icalp}[1]{Proceedings of the #1
  International Colloquium on Automata, Languages, and Programming
  ({ICALP})}\newcommand{\itcs}[1]{Proceedings of the #1 Innovations in
  Theoretical Computer Science Conference
  (ITCS)}\newcommand{\mfcs}[1]{Proceedings of the #1 International Symposium on
  Mathematical Foundations of Computer Science (MFCS)}\newcommand{\toc}{Theory
  of Computing}\newcommand{\sicomp}{SIAM Journal on
  Computing}\newcommand{\jacm}{Journal of the ACM}
\begin{thebibliography}{GKMV21}

\bibitem[Abb19]{dp_abboud_2019}
Amir Abboud.
\newblock Fine-grained reductions and quantum speedups for dynamic programming.
\newblock In {\em \icalp{46th}}, volume 132 of {\em Leibniz International
  Proceedings in Informatics (LIPIcs)}, pages 8:1--8:13, 2019.
\newblock \href {https://doi.org/10.4230/LIPIcs.ICALP.2019.8}
  {\path{doi:10.4230/LIPIcs.ICALP.2019.8}}.

\bibitem[ABI{\etalchar{+}}19]{dp_ambainis_2019}
Andris Ambainis, Kaspars Balodis, J{\=a}nis Iraids, Martins Kokainis, Kri{\v
  s}j{\=a}nis Pr{\=u}sis, and Jevg{\=e}nijs Vihrovs.
\newblock Quantum speedups for exponential-time dynamic programming algorithms.
\newblock In {\em \soda{2019}}, pages 1783--1793, 2019.
\newblock \href {http://arxiv.org/abs/1807.05209} {\path{arXiv:1807.05209}},
  \href {https://doi.org/10.1137/1.9781611975482.107}
  {\path{doi:10.1137/1.9781611975482.107}}.

\bibitem[ABI{\etalchar{+}}20]{2d_grid_ambainis_2020}
Andris Ambainis, Kaspars Balodis, Jānis Iraids, Kamil Khadiev, Vladislavs
  Kļevickis, Kri{\v{s}}jānis Prūsis, Yixin Shen, Juris Smotrovs, and
  Jevgēnijs Vihrovs.
\newblock Quantum lower and upper bounds for {$2D$}-grid and {D}yck language.
\newblock In {\em \mfcs{45th}}, volume 170 of {\em Leibniz International
  Proceedings in Informatics (LIPIcs)}, pages 8:1--8:14. Schloss
  Dagstuhl--Leibniz-Zentrum f{\"u}r Informatik, 2020.
\newblock \href {http://arxiv.org/abs/2007.03402} {\path{arXiv:2007.03402}},
  \href {https://doi.org/10.4230/LIPIcs.MFCS.2020.8}
  {\path{doi:10.4230/LIPIcs.MFCS.2020.8}}.

\bibitem[ABW15]{lcs_abboud_2015}
Amir Abboud, Arturs Backurs, and Virginia~Vassilevska Williams.
\newblock Tight hardness results for {LCS} and other sequence similarity
  measures.
\newblock In {\em \focs{56th}}, pages 59--78, 2015.
\newblock \href {http://arxiv.org/abs/1501.07053} {\path{arXiv:1501.07053}},
  \href {https://doi.org/10.1109/FOCS.2015.14}
  {\path{doi:10.1109/FOCS.2015.14}}.

\bibitem[AD99]{lis_aldous_1999}
David Aldous and Persi Diaconis.
\newblock Longest increasing subsequences: From patience sorting to the
  {Baik-Deift-Johansson} theorem.
\newblock {\em Bulletin of the American Mathematical Society}, 36:413--432,
  1999.
\newblock \href {https://doi.org/10.1090/s0273-0979-99-00796-x}
  {\path{doi:10.1090/s0273-0979-99-00796-x}}.

\bibitem[AGS19]{trichotomy_aaronson_2019}
Scott Aaronson, Daniel Grier, and Luke Schaeffer.
\newblock A quantum query complexity trichotomy for regular languages.
\newblock In {\em \focs{60th}}, pages 942--965, 2019.
\newblock \href {http://arxiv.org/abs/1812.04219} {\path{arXiv:1812.04219}},
  \href {https://doi.org/10.1109/FOCS.2019.00061}
  {\path{doi:10.1109/FOCS.2019.00061}}.

\bibitem[AJ22]{string_akmaljin_2022}
Shyan Akmal and Ce~Jin.
\newblock Near-optimal quantum algorithms for string problems.
\newblock In {\em \soda{2022}}, pages 2791--2832, 2022.
\newblock \href {http://arxiv.org/abs/2110.09696} {\path{arXiv:2110.09696}},
  \href {https://doi.org/10.1137/1.9781611977073.109}
  {\path{doi:10.1137/1.9781611977073.109}}.

\bibitem[AKO10]{lcs_asymmetric_andoni_2010}
Alexandr Andoni, Robert Krauthgamer, and Krzysztof Onak.
\newblock Polylogarithmic approximation for edit distance and the asymmetric
  query complexity.
\newblock In {\em \focs{51st}}, pages 377--386, 2010.
\newblock \href {http://arxiv.org/abs/1005.4033} {\path{arXiv:1005.4033}},
  \href {https://doi.org/10.1109/FOCS.2010.43}
  {\path{doi:10.1109/FOCS.2010.43}}.

\bibitem[Amb07]{qwalk_elem_distictness_ambainis_2007}
Andris Ambainis.
\newblock Quantum walk algorithm for element distinctness.
\newblock {\em SIAM Journal on Computing}, 37(1):210--239, 2007.
\newblock \href {http://arxiv.org/abs/quant-ph/0311001}
  {\path{arXiv:quant-ph/0311001}}, \href
  {https://doi.org/10.1137/S0097539705447311}
  {\path{doi:10.1137/S0097539705447311}}.

\bibitem[AS04]{collision_aaronson_shi_2004}
Scott Aaronson and Yaoyun Shi.
\newblock Quantum lower bounds for the collision and the element distinctness
  problems.
\newblock {\em \jacm}, 51(4):595–605, 2004.
\newblock \href {https://doi.org/10.1145/1008731.1008735}
  {\path{doi:10.1145/1008731.1008735}}.

\bibitem[AVW21]{lcs_akmal_2021}
Shyan Akmal and Virginia Vassilevska~Williams.
\newblock Improved approximation for longest common subsequence over small
  alphabets.
\newblock In {\em \icalp{48th}}, volume 198 of {\em Leibniz International
  Proceedings in Informatics (LIPIcs)}, pages 13:1--13:18, 2021.
\newblock \href {http://arxiv.org/abs/2105.03028} {\path{arXiv:2105.03028}},
  \href {https://doi.org/10.4230/LIPIcs.ICALP.2021.13}
  {\path{doi:10.4230/LIPIcs.ICALP.2021.13}}.

\bibitem[AW21]{matrix_williams_2021}
Josh Alman and Virginia~Vassilevska Williams.
\newblock A refined laser method and faster matrix multiplication.
\newblock In {\em \soda{2021}}, pages 522--539, 2021.
\newblock \href {http://arxiv.org/abs/2010.05846} {\path{arXiv:2010.05846}},
  \href {https://doi.org/10.1137/1.9781611976465.32}
  {\path{doi:10.1137/1.9781611976465.32}}.

\bibitem[BBC{\etalchar{+}}01]{lower_bbcmw_2001}
Robert Beals, Harry Buhrman, Richard Cleve, Michele Mosca, and Ronald de~Wolf.
\newblock Quantum lower bounds by polynomials.
\newblock {\em \jacm}, 48(4):778--797, 2001.
\newblock \href {http://arxiv.org/abs/quant-ph/9802049}
  {\path{arXiv:quant-ph/9802049}}.

\bibitem[Bel12a]{kdist_belovs_2012}
Aleksandrs Belovs.
\newblock Learning-graph-based quantum algorithm for $k$-distinctness.
\newblock In {\em \focs{53rd}}, page 207–216, 2012.
\newblock \href {http://arxiv.org/abs/1205.1534} {\path{arXiv:1205.1534}},
  \href {https://doi.org/10.1109/FOCS.2012.18}
  {\path{doi:10.1109/FOCS.2012.18}}.

\bibitem[Bel12b]{span1_belovs_2012}
Aleksandrs Belovs.
\newblock Span programs for functions with constant-sized 1-certificates.
\newblock In {\em \stoc{44th}}, page 77–84, 2012.
\newblock \href {http://arxiv.org/abs/1105.4024} {\path{arXiv:1105.4024}},
  \href {https://doi.org/10.1145/2213977.2213985}
  {\path{doi:10.1145/2213977.2213985}}.

\bibitem[Bel14]{juntas_belovs_2014}
Aleksandrs Belovs.
\newblock Quantum algorithms for learning symmetric juntas via adversary bound.
\newblock In {\em Proceedings of the 29th IEEE Conference on Computational
  Complexity (CCC)}, pages 22--31, 2014.
\newblock \href {http://arxiv.org/abs/1311.6777} {\path{arXiv:1311.6777}},
  \href {https://doi.org/10.1109/CCC.2014.11} {\path{doi:10.1109/CCC.2014.11}}.

\bibitem[BHMT02]{amplitude_brassard_2002}
Gilles Brassard, Peter H{\o}yer, Michele Mosca, and Alain Tapp.
\newblock Quantum amplitude amplification and estimation.
\newblock {\em Contemporary Mathematics}, 305:53--74, 2002.
\newblock \href {http://arxiv.org/abs/quant-ph/0005055}
  {\path{arXiv:quant-ph/0005055}}, \href
  {https://doi.org/10.1090/conm/305/05215} {\path{doi:10.1090/conm/305/05215}}.

\bibitem[BHS80]{master_theorem_bentley_haken_saxe_1980}
Jon~Louis Bentley, Dorothea Haken, and James~B. Saxe.
\newblock A general method for solving divide-and-conquer recurrences.
\newblock {\em SIGACT News}, 12(3):36–44, 1980.
\newblock \href {https://doi.org/10.1145/1008861.1008865}
  {\path{doi:10.1145/1008861.1008865}}.

\bibitem[BR12]{st_belovs_2012}
Aleksandrs Belovs and Ben~W. Reichardt.
\newblock Span programs and quantum algorithms for $st$-connectivity and claw
  detection.
\newblock In {\em Algorithms -- ESA 2012}, pages 193--204, 2012.
\newblock \href {http://arxiv.org/abs/1203.2603} {\path{arXiv:1203.2603}},
  \href {https://doi.org/10.1007/978-3-642-33090-2_18}
  {\path{doi:10.1007/978-3-642-33090-2_18}}.

\bibitem[CJOP20]{span_cornelissen_2020}
Arjan Cornelissen, Stacey Jeffery, Maris Ozols, and Alvaro Piedrafita.
\newblock Span programs and quantum time complexity.
\newblock In {\em \mfcs{45th}}, volume 170 of {\em Leibniz International
  Proceedings in Informatics (LIPIcs)}, pages 26:1--26:14, 2020.
\newblock \href {http://arxiv.org/abs/2005.01323} {\path{arXiv:2005.01323}},
  \href {https://doi.org/10.4230/LIPIcs.MFCS.2020.26}
  {\path{doi:10.4230/LIPIcs.MFCS.2020.26}}.

\bibitem[CLRS09]{algorithms_cormen_2009}
Thomas~H. Cormen, Charles~E. Leiserson, Ronald~L. Rivest, and Clifford Stein.
\newblock {\em Introduction to Algorithms}.
\newblock MIT Press, third edition, 2009.

\bibitem[CT65]{fourier_cooley_tukey_1965}
James~W. Cooley and John~W. Tukey.
\newblock An algorithm for the machine calculation of complex {F}ourier series.
\newblock {\em Mathematics of Computation}, 19(90):297--301, 1965.
\newblock \href {https://doi.org/10.2307/2003354} {\path{doi:10.2307/2003354}}.

\bibitem[CW90]{matrix_coppersmith_winograd_1990}
Don Coppersmith and Shmuel Winograd.
\newblock Matrix multiplication via arithmetic progressions.
\newblock {\em Journal of Symbolic Computation}, 9(3):251--280, 1990.
\newblock \href {https://doi.org/10.1016/S0747-7171(08)80013-2}
  {\path{doi:10.1016/S0747-7171(08)80013-2}}.

\bibitem[DH96]{minimum_finding_durr_hoyer_1996}
Christoph D{\"u}rr and Peter H{\o}yer.
\newblock A quantum algorithm for finding the minimum, 1996.
\newblock \href {http://arxiv.org/abs/quant-ph/9607014}
  {\path{arXiv:quant-ph/9607014}}.

\bibitem[Fre75]{lis_fredman_1975}
Michael~L. Fredman.
\newblock On computing the length of longest increasing subsequences.
\newblock {\em Discrete Mathematics}, 11(1):29--35, 1975.
\newblock \href {https://doi.org/10.1016/0012-365X(75)90103-X}
  {\path{doi:10.1016/0012-365X(75)90103-X}}.

\bibitem[FRPU94]{noisybinary_feige_1994}
Uriel Feige, Prabhakar Raghavan, David Peleg, and Eli Upfal.
\newblock Computing with noisy information.
\newblock {\em SIAM Journal on Computing}, 23(5):1001--1018, 1994.
\newblock \href {https://doi.org/10.1137/S0097539791195877}
  {\path{doi:10.1137/S0097539791195877}}.

\bibitem[GKMV21]{dp_glos_2021}
Adam Glos, Martins Kokainis, Ryuhei Mori, and Jevg\={e}nijs Vihrovs.
\newblock Quantum speedups for dynamic programming on $n$-dimensional lattice
  graphs.
\newblock In {\em \mfcs{46th}}, volume 202 of {\em Leibniz International
  Proceedings in Informatics (LIPIcs)}, pages 50:1--50:23, 2021.
\newblock \href {http://arxiv.org/abs/2104.14384} {\path{arXiv:2104.14384}},
  \href {https://doi.org/10.4230/LIPIcs.MFCS.2021.50}
  {\path{doi:10.4230/LIPIcs.MFCS.2021.50}}.

\bibitem[Gro96]{search_grover_1996}
Lov~K. Grover.
\newblock A fast quantum mechanical algorithm for database search.
\newblock In {\em \stoc{28th}}, page 212–219, 1996.
\newblock \href {http://arxiv.org/abs/quant-ph/9605043}
  {\path{arXiv:quant-ph/9605043}}, \href
  {https://doi.org/10.1145/237814.237866} {\path{doi:10.1145/237814.237866}}.

\bibitem[HL{\v S}07]{adversary_hoyer_2007}
Peter H{\o}yer, Troy Lee, and Robert {\v S}palek.
\newblock Negative weights make adversaries stronger.
\newblock In {\em \stoc{39th}}, page 526–535, 2007.
\newblock \href {http://arxiv.org/abs/quant-ph/0611054}
  {\path{arXiv:quant-ph/0611054}}, \href
  {https://doi.org/10.1145/1250790.1250867}
  {\path{doi:10.1145/1250790.1250867}}.

\bibitem[Hoa62]{quicksort_hoare_1962}
C.~A.~R. Hoare.
\newblock Quicksort.
\newblock {\em The Computer Journal}, 5(1):10--16, 1962.
\newblock \href {https://doi.org/10.1093/comjnl/5.1.10}
  {\path{doi:10.1093/comjnl/5.1.10}}.

\bibitem[Jef22]{span_jeffery_2022}
Stacey Jeffery.
\newblock Span programs and quantum space complexity.
\newblock {\em \toc}, 18(11):1--49, 2022.
\newblock \href {http://arxiv.org/abs/1908.04232} {\path{arXiv:1908.04232}},
  \href {https://doi.org/10.4086/toc.2022.v018a011}
  {\path{doi:10.4086/toc.2022.v018a011}}.

\bibitem[Kim13]{adversaryupper_kimmel_2013}
Shelby Kimmel.
\newblock Quantum adversary (upper) bound.
\newblock {\em Chicago Journal of Theoretical Computer Science}, 2013(4), 2013.
\newblock \href {http://arxiv.org/abs/1101.0797} {\path{arXiv:1101.0797}},
  \href {https://doi.org/10.4086/cjtcs.2013.004}
  {\path{doi:10.4086/cjtcs.2013.004}}.

\bibitem[KO63]{multiplication_karatsuba_1963}
Anatolij~A. Karatsuba and Yu. Ofman.
\newblock Multiplication of multidigit numbers on automata.
\newblock {\em Soviet Physics Doklady}, 7:595--596, 1963.

\bibitem[KPV22]{dp_klevickis_2022}
Vladislavs K\c{l}evickis, Kri\v{s}j\={a}nis Pr\={u}sis, and Jevg\={e}nijs
  Vihrovs.
\newblock Quantum speedups for treewidth.
\newblock In {\em 17th Conference on the Theory of Quantum Computation,
  Communication and Cryptography (TQC 2022)}, volume 232 of {\em Leibniz
  International Proceedings in Informatics (LIPIcs)}, pages 11:1--11:18.
  Schloss Dagstuhl -- Leibniz-Zentrum f{\"u}r Informatik, 2022.
\newblock \href {http://arxiv.org/abs/2202.08186} {\path{arXiv:2202.08186}},
  \href {https://doi.org/10.4230/LIPIcs.TQC.2022.11}
  {\path{doi:10.4230/LIPIcs.TQC.2022.11}}.

\bibitem[LMR{\etalchar{+}}11]{stateconversion_lee_2011}
Troy Lee, Rajat Mittal, Ben~W. Reichardt, Robert Špalek, and Mario Szegedy.
\newblock Quantum query complexity of state conversion.
\newblock In {\em \focs{52nd}}, pages 344--353, 2011.
\newblock \href {http://arxiv.org/abs/1011.3020} {\path{arXiv:1011.3020}},
  \href {https://doi.org/10.1109/FOCS.2011.75}
  {\path{doi:10.1109/FOCS.2011.75}}.

\bibitem[LMR{\v S}10]{stateconversion_lee_2011_v1}
Troy Lee, Rajat Mittal, Ben~W. Reichardt, and Robert {\v S}palek.
\newblock An adversary for algorithms, 2010.
\newblock \href {http://arxiv.org/abs/1011.3020v1} {\path{arXiv:1011.3020v1}}.

\bibitem[LMS98]{lcs_landau_1998}
Gad~M. Landau, Eugene~W. Myers, and Jeanette~P. Schmidt.
\newblock Incremental string comparison.
\newblock {\em \sicomp}, 27(2):557--582, 1998.
\newblock \href {https://doi.org/10.1137/S0097539794264810}
  {\path{doi:10.1137/S0097539794264810}}.

\bibitem[LMS17]{triangle_lee_2017}
Troy Lee, Fr{\'e}d{\'e}ric Magniez, and Miklos Santha.
\newblock Improved quantum query algorithms for triangle detection and
  associativity testing.
\newblock {\em Algorithmica}, 77(2):459--486, 2017.
\newblock \href {http://arxiv.org/abs/1210.1014} {\path{arXiv:1210.1014}},
  \href {https://doi.org/10.1007/s00453-015-0084-9}
  {\path{doi:10.1007/s00453-015-0084-9}}.

\bibitem[Rei09a]{unbalanced_reichardt_2009}
Ben~W. Reichardt.
\newblock Span-program-based quantum algorithm for evaluating unbalanced
  formulas, 2009.
\newblock \href {http://arxiv.org/abs/0907.1622} {\path{arXiv:0907.1622}}.

\bibitem[Rei09b]{span_reichardt_2009}
Ben~W. Reichardt.
\newblock Span programs and quantum query complexity: The general adversary
  bound is nearly tight for every {B}oolean function.
\newblock In {\em \focs{50th}}, pages 544--551, 2009.
\newblock \href {http://arxiv.org/abs/0904.2759} {\path{arXiv:0904.2759}},
  \href {https://doi.org/10.1109/FOCS.2009.55}
  {\path{doi:10.1109/FOCS.2009.55}}.

\bibitem[Rei11]{reflections_reichardt_2011}
Ben~W. Reichardt.
\newblock Reflections for quantum query algorithms.
\newblock In {\em \soda{22nd}}, page 560–569, 2011.
\newblock \href {http://arxiv.org/abs/1005.1601} {\path{arXiv:1005.1601}},
  \href {https://doi.org/10.5555/2133036.2133080}
  {\path{doi:10.5555/2133036.2133080}}.

\bibitem[R{\v S}12]{formulas_reichardt_2012}
Ben Reichardt and Robert {\v S}palek.
\newblock Span-program-based quantum algorithm for evaluating formulas.
\newblock {\em \toc}, 8(13):291--319, 2012.
\newblock \href {http://arxiv.org/abs/0710.2630} {\path{arXiv:0710.2630}},
  \href {https://doi.org/10.4086/toc.2012.v008a013}
  {\path{doi:10.4086/toc.2012.v008a013}}.

\bibitem[RSSS19]{lcs_lis_rubinstein_2019}
Aviad Rubinstein, Saeed Seddighin, Zhao Song, and Xiaorui Sun.
\newblock Approximation algorithms for {LCS} and {LIS} with truly improved
  running times.
\newblock In {\em \focs{60th}}, pages 1121--1145, 2019.
\newblock \href {http://arxiv.org/abs/2111.10538} {\path{arXiv:2111.10538}},
  \href {https://doi.org/10.1109/FOCS.2019.00071}
  {\path{doi:10.1109/FOCS.2019.00071}}.

\bibitem[RV03]{exact_string_matching_ramesh_vinay_2003}
H.~Ramesh and V.~Vinay.
\newblock String matching in {$\tilde O(\sqrt{n}+\sqrt{m})$} quantum time.
\newblock {\em Journal of Discrete Algorithms}, 1(1):103--110, 2003.
\newblock \href {http://arxiv.org/abs/quant-ph/0011049}
  {\path{arXiv:quant-ph/0011049}}, \href
  {https://doi.org/10.1016/S1570-8667(03)00010-8}
  {\path{doi:10.1016/S1570-8667(03)00010-8}}.

\bibitem[SS10]{lis_saks_2010}
Michael Saks and C.~Seshadhri.
\newblock Estimating the longest increasing sequence in polylogarithmic time.
\newblock In {\em \focs{51st}}, pages 458--467, 2010.
\newblock \href {http://arxiv.org/abs/1308.0626} {\path{arXiv:1308.0626}},
  \href {https://doi.org/10.1109/FOCS.2010.51}
  {\path{doi:10.1109/FOCS.2010.51}}.

\bibitem[Str69]{matrix_strassen_1969}
Volker Strassen.
\newblock Gaussian elimination is not optimal.
\newblock {\em Numerische Mathematik}, 13(4):354--356, 1969.
\newblock \href {https://doi.org/10.1007/BF02165411}
  {\path{doi:10.1007/BF02165411}}.

\bibitem[WF74]{lcs_wagner_1974}
Robert~A. Wagner and Michael~J. Fischer.
\newblock The string-to-string correction problem.
\newblock {\em \jacm}, 21(1):168–173, 1974.
\newblock \href {https://doi.org/10.1145/321796.321811}
  {\path{doi:10.1145/321796.321811}}.

\end{thebibliography}
\bibliographystyle{alphaurl}

\appendix
\section{Appendix}\label{sec:appendix}


\subsection{Adversary composition lemmas}\label{app:adv-composition}

In this appendix, we prove two adversary compositions lemmas: \Cref{lem:and-or} in \Cref{app:and-or} and \Cref{lem:adv-switch} in \Cref{app:adv-switch}. Before proving these lemmas, we first introduce a quantity closely related to $\adv$, the filtered $\gammatwo$ norm, and state some facts about $\adv$ and $\gammatwo$ that are used in our proofs.

For finite sets $D_1$ and $D_2$, we write $A\in \mathbb{C}^{D_1\times D_2}$ to mean that $A$ is a complex $|D_1| \times |D_2|$ matrix with rows indexed by $D_1$ and columns indexed by $D_2$. 
\begin{Definition}[Filtered $\gammatwo$ norm] \label{def:gamma}
Let $D_1$ and $D_2$ be finite sets and $n\in \mathbb{N}$. Let $A \in \mathbb{C}^{D_1\times D_2}$ and $Z = \{Z_i \in \mathbb{C}^{D_1 \times D_2} \mid i \in [n]\}$. Define $\gammatwo (A | Z)$ by
\begin{alignat}{2}
    &\gammatwo(A | Z)\coloneqq  &&\min_{\substack{d\in \mathbb{N},\\[1pt]\ket{u_{xj}},\ket{v_{yj}}\in \mathbb{C}^d}}  \; \max \left\{ \max_{x\in D_1} \sum_{j=1}^n \norm{\ket{u_{xj}}}^2, \max_{y\in D_2} \sum_{j=1}^n \norm{\ket{v_{yj}}}^2 \right\} \\
    &\ \textup{subject to:} &&\quad \forall x \in D_1, y \in D_2, \quad \sum_{j:x_j \neq y_j} (Z_j)_{xy}  \braket{u_{xj}}{v_{xj}} = A_{xy}. \label{eq:gamma2_last}
\end{alignat}  
\end{Definition}

We can view $\adv(\cdot)$ as a special case of the filtered $\gammatwo$ norm. Let $f\colon D \to E$ be a function, where $D\subseteq \mathcal{D}^n$ and $E$ are finite sets. Let $F$ be the Gram matrix of $f$. Let $\Delta \coloneqq \{\Delta_1, \ldots, \Delta_n\}$ be a set of $|D|\times |D|$ matrices defined such that $(\Delta_j)_{x, y \in D} = 1 - \delta_{x_j,y_j}$ for all $j\in [n]$. Then it can be shown that $\adv(f) = \gammatwo(J - F | \Delta')$, where $\Delta' \coloneqq \{\Delta_j \circ (J - F) \mid j \in [n]\}$ and $J$ is the all-ones matrix of the same size as $F$. 

We now state five facts about $\adv$ and $\gammatwo$ that are used in our proofs. The symbols $\Delta$ and $J$ refer to the set of matrices defined above (appropriately sized) and the all-ones matrix (appropriately sized), respectively.

\begin{Fact}[{\cite[Theorem 6.2]{span_reichardt_2009}}]\label{fact:adv-optimization} Let $f\colon D \to \{0, 1\}$ be a Boolean function, where $D \subseteq \mathcal{D}^n$ is a finite set. Then
\begin{alignat}{2}
    &\adv(f) = &&\min_{\substack{d\in \mathbb{N},\\[1pt]\ket{u_{xj}}\in \mathbb{C}^d}} \sqrt{A_0\cdot A_1}
    \\
    &\ \textup{subject to:} &&\quad
    \forall b\in \{0,1\}, \quad A_b = \max_{x\in f^{-1}(b)} \sum_{j=1}^n \norm{\ket{u_{xj}}}^2.
    \\
    & &&\quad \forall x,y\in D \textup{ with } f(x)\neq f(y), \quad \sum_{j \in [n]: \, x_j \neq y_j}  \braket{u_{xj}}{u_{yj}} = 1.
\end{alignat}
\end{Fact}

\begin{Fact}[{\cite[Theorem 3.4]{stateconversion_lee_2011}}]\label{fact:gamma-bool}
Let $f\colon D \to E$ be a function, where $D\subseteq \mathcal{D}^n$ and $E$ are finite sets. Let $F$ be the Gram matrix of $f$. Then 
\begin{equation}
\adv(f) \leq \gammatwo(J - F | \Delta) \leq 2(1 - 1/|E|) \adv(f).
\end{equation}
\end{Fact}

\begin{Fact}[Triangle inequality~{\cite[Lemma A.2, item 4]{stateconversion_lee_2011}}]
\label{fact:gamma-triangle}
Let $D_1, D_2$ be finite sets and $n\in \mathbb{N}$. Let $A, B \in \mathbb{C}^{D_1 \times D_2}$ and $Y_j \in \mathbb{C}^{D_1 \times D_2}$ for all $j\in [n]$. Let $Z \coloneqq \{Z_j \mid j\in [n]\}$. Then
\begin{equation}
    \gammatwo(A + B | Z) \leq \gammatwo(A | Z) + \gammatwo(B | Z).
\end{equation}
\end{Fact}

\begin{Fact}[{\cite[Lemma A.2, item 10]{stateconversion_lee_2011}}]\label{fact:gamma-hadamard}
Let $D_1, D_2$ be finite sets and $n\in \mathbb{N}$. Let $A, B \in \mathbb{C}^{D_1\times D_2}$ and $Y_j, Z_j \in \mathbb{C}^{D_1\times D_2}$ for all $j\in [n]$. Let $Y \coloneqq \{Y_j \mid j\in [n]\}$ and $Z \coloneqq \{Z_j \mid j \in [n]\}$. Then
\begin{equation}
\gammatwo (A \circ B \mid \{Z_j \circ B \mid j\in [n]\}) \leq \gammatwo(A \mid Z) \leq \gammatwo(A \mid \{Z_j \circ B \mid j \in [n]\}) \cdot \gammatwo(B).
\end{equation}
\end{Fact}

\begin{Fact}[{Direct-sum property \cite[Lemma A.2, item 12]{stateconversion_lee_2011}}]\label{fact:gamma-sum}
Let $C_1, C_2, D_1, D_2$ be finite sets and $n\in \mathbb{N}$. Let $A \in \mathbb{C}^{C_1\times C_2}$ and $Y_j \in \mathbb{C}^{C_1\times C_2}$ for all $j\in [n]$. Let $B \in \mathbb{C}^{D_1\times D_2}$ and $Z_j \in \mathbb{C}^{D_1\times D_2}$ for all $j\in [n]$. Let $Y \coloneqq \{Y_j \mid j\in [n]\}$ and $Z \coloneqq \{Z_j \mid j \in [n]\}$. Then
\begin{equation}
\gammatwo(A \oplus B \mid \{Y_j \oplus Z_j \mid j \in [n]\}) = \max \{\gammatwo(A | Y), \gammatwo (B | Z)\}.
\end{equation}
\end{Fact}

\subsubsection{Proof of \texorpdfstring{\Cref{lem:and-or}}{Lemma \ref{lem:and-or}}}
\label{app:and-or}

\lemAdversaryAndOr*

\begin{proof}
It suffices to prove the lemma only for $g^\vee$ since  $g^\wedge(x,y) = f_1(x)\wedge f_2(y) = \lnot((\lnot f_1(x))\vee(\lnot f_2(x)))$ and $\adv(f) = \adv(\lnot f)$ for any Boolean function $f\colon \Sigma^n \to \{0,1\}$, where $\lnot$ denotes negation. For the rest of this proof, we write $g \coloneqq g^\vee$, $a_1 \coloneqq a$, and $a_2 \coloneqq b$ for notational convenience.

For $i\in \{1,2\}$, let $\bigl\{\ket{u_{xj}^i} \in \mathbb{C}^{d_i} \mid x\in \Sigma^{a_i}, j \in  [a_i]\bigr\}$ be a minimum solution to the minimization problem in \Cref{fact:adv-optimization} that specifies $\adv(f_i)$. Then
\begin{align}
    &\adv(f_i)^2 = A_0^i \cdot A_1^i, \quad \text{where $A_b^i \coloneqq \max_{x\in f_i^{-1}(b)} \sum_{j=1}^{a_i} \norm{\ket{u_{xj}^i}}^2$ for all $b\in \{0,1\}$,}
    \\
    &\ \text{$\forall x,y\in \Sigma^{a_i}$ with $f_i(x)\neq f_i(y)$,} \quad\sum_{j\in [a_i]:\, x_j\neq y_j}\braket{u_{xj}^i}{u_{yj}^i} = 1.
\end{align}
By mapping $\ket{u_{xj}^i} \mapsto \sqrt{A^i_1} \,  \ket{u_{xj}^i}$ for $x\in f_i^{-1}(0)$ and $\ket{u_{xj}^i} \mapsto \ket{u_{xj}^i}/ \sqrt{A_1^i}$ for $x\in f_i^{-1}(1)$, we can assume without loss of generality that
\begin{equation}\label{eq:rescale}
    A_0^i = \adv(f_i)^2 \quad \text{and} \quad  A_1^i = 1 \quad \text{for all $i\in \{1,2\}$}.
\end{equation}

Now, for $(x,x')\in \Sigma^{a_1} \times \Sigma^{a_2}$ and $k\in [a_1+a_2]$, let $k' \coloneqq k - a_1$ and define $\ket{\lambda_{(x,x')k}} \in \mathbb{C}^{d_1}\oplus \mathbb{C}^{d_2}$ according to the following table. For example, if $(f_1(x),f_2(x'))=(1,0)$ and $k\leq a_1$, then $\ket{\lambda_{(x,x')k}} \coloneqq \ket{u_{xk}^1} \oplus 0$.
\begin{equation}
\renewcommand{\arraystretch}{1.6}
\begin{tabular}{c|c|c}
     $(f_1(x),f_2(x'))$ & $k\leq a_1$ & $k>a_1$ \\[1pt]
    \hline
     $(0,0)$ & $\ket{u_{xk}^1}\oplus 0$  & $0 \oplus \ket{u_{x'k'}^2}$ \\
     $(0,1)$ & $0 \oplus 0$  & $0 \oplus \ket{u_{x'k'}^2}$ \\
     $(1,0)$ & $\ket{u_{xk}^1} \oplus 0$  & $0 \oplus 0$ \\
     $(1,1)$ & $\ket{u_{xk}^1} \oplus 0$  & $0 \oplus 0$
\end{tabular}
\end{equation}

Then it can be directly verified that
\begin{equation}\label{eq:constraint_satisfied}
    \text{$\forall (x,x'),(y,y')\in \Sigma^{a_1}\times \Sigma^{a_2}$ with $g(x,x')\neq g(y,y')$,} \quad \sum_{\substack{k \in [a_1 + a_2]: \\[1pt] (x,x')_k \neq (y,y')_k}} \braket{\lambda_{(x,x') k}}{\lambda_{(y,y') k}} = 1.
\end{equation}
Moreover, from the definitions of $g$ and $\ket{\lambda_{(x,x')k}}$, and \Cref{eq:rescale}, we find
\begin{equation}
\begin{aligned}
   A_0 &\coloneqq \max_{(x,x')\in g^{-1}(0)} \sum_{k=1}^{a_1+a_2} \norm{\ket{\lambda_{(x,x')k}}}^2  = \adv(f_1)^2 + \adv(f_2)^2,\\
   A_1 &\coloneqq \max_{(x,x')\in g^{-1}(1)} \sum_{k=1}^{a_1+a_2} \norm{\ket{\lambda_{(x,x')k}}}^2  = 1.
\end{aligned}
\end{equation}

But \Cref{eq:constraint_satisfied} means that $\bigl\{\ket{\lambda_{(x,x')k} \in \mathbb{C}^{d_1}\oplus \mathbb{C}^{d_2} \mid (x,x') \in \Sigma^{a_1}\times \Sigma^{a_2}, k \in [a_1+a_2]}\bigr\}$ is a feasible solution to the minimization problem in \Cref{fact:adv-optimization} that specifies $\adv(g)$. Therefore
\begin{equation}
    \adv(g)^2 \leq A_0\cdot A_1 = \adv(f_1)^2 + \adv(f_2)^2,
\end{equation}
which completes the proof.
\end{proof}

\subsubsection{Proof of \texorpdfstring{\Cref{lem:adv-switch}}{Lemma \ref{lem:adv-switch}}}
\label{app:adv-switch}

\lemAdversarySwitch*

\begin{proof}
Define $h'\colon \Sigma^a \to \Lambda \times \{0,1\}$ by $h'(x) = (f(x), h(x)) = (f(x), g_{f(x)}(x))$. Let $H'$, $F$, and $\{G_s \mid s\in \Lambda\}$ be the Gram matrices for $h'$, $f$, and $\{g_s \mid s\in \Lambda\}$, respectively, each having size $|\Sigma|^a \times |\Sigma|^a$. Define $\Delta \coloneqq \{\Delta_1,\ldots,\Delta_a\}$, where each $\Delta_j$ is a $|\Sigma|^a \times |\Sigma|^a$ matrix with entries $(\Delta_j)_{xy} \coloneqq 1-\delta_{x_j,y_j}$ for all $x,y \in \Sigma^a$. For $m\in \mathbb{N}$, define $J_{m}$ to be the all-ones matrix of size $m \times m$, and write $J \coloneqq J_{|\Sigma|^a}$. We also write $\gammatwo(f) \coloneqq \gammatwo(J-F|\Delta)$.

We have
\begin{align}
    \adv(h)
    &\leq \adv(h') &&(\text{\Cref{def:adv}})
    \\
    &\leq \gammatwo(J - H' \mid \Delta) &&(\text{\Cref{fact:gamma-bool}})
    \\
    &\leq \gammatwo(J - F \mid \Delta) + \gammatwo(F - H' \mid \Delta) &&\left(\text{\Cref{fact:gamma-triangle}}\right)
    \\
    &\leq O(\adv(f)) + \gammatwo(F - H' \mid \Delta)
    &&(\text{\Cref{fact:gamma-bool}}).
\end{align}

We proceed to bound $\gammatwo(F - H' \mid \Delta)$. For $s\in \Lambda$, let $b_s\coloneqq |f^{-1}(s)|$, where $f^{-1}(s)\coloneqq \{x\in \Sigma^a \mid f(x) = s\}$. Note $\sum_{s\in \Lambda} b_s = |\Sigma|^a$. For $s\in \Lambda$, define a matrix $\tilde{G}_s$ of size $b_s\times b_s$ by $(\tilde{G}_s)_{xy} = (G_s)_{xy}$ for all $x, y \in f^{-1}(s)$. For $s\in \Lambda$ and $j\in [a]$, define a matrix $\Delta_j^{(b_s)}$ of size $b_s\times b_s$ by $\bigl(\Delta_j^{(b_s)}\bigr)_{xy} = 1-\delta_{x_j,y_j}$ for all $x, y \in f^{-1}(s)$. Define $\Delta^{(b_s)}\coloneqq \{\Delta_1^{(b_s)},\ldots, \Delta_a^{(b_s)}\}$.

By inspection, we have
\begin{equation}\label{eq:block-decomposition}
    H' = \bigoplus_{s\in \Lambda} \tilde{G}_{s}, \quad F = \bigoplus_{s\in \Lambda} J_{b_s}, \quad \text{and} \quad  \Delta_j \circ F = \bigoplus_{s \in \Lambda} \Delta_j^{(b_s)}.
\end{equation}

Therefore, we have
\begin{align}
    \gammatwo(F-H' \mid \Delta)  
    &\leq \gammatwo(F-H' \mid \{\Delta_j \circ F \mid j \in [a]\}) \cdot \gammatwo(F) &&(\text{\Cref{fact:gamma-hadamard}})
    \\
    &\leq \gammatwo(F-H' \mid \{\Delta_j \circ F \mid j \in [a]\}) &&(\text{$\gammatwo(F) \leq  1$})
    \\
    &= \max_{s\in \Lambda} \gammatwo(J_{b_s} - \tilde{G}_s \mid  \Delta^{(b_s)})  &&(\text{\Cref{eq:block-decomposition} and \Cref{fact:gamma-sum}})
    \\
    &\leq \max_{s \in \Lambda} \gammatwo(J - G_s \mid \Delta) = \max_{s \in \Lambda} \adv(g_s).
\end{align}
The lemma follows.
\end{proof}

\subsection{Randomized search}\label{app:randomized-reduction}
In this appendix, we equate the totally ordered set $\Sigma$ with $\mathbb{Z}$ for convenience of presentation. It is clear that our randomized search can be generalized to any totally ordered set $\Sigma$.

We claim that $\mt_{j,n}$ and $\mh_{j,n}$ can be computed by a randomized search that uses $O(\log(n))$ computations of $\LIS_{j,n}$ and calls to Grover search. The claim follows from \Cref{lemma:randomized_search}, which describes the randomized search, with the oracles $\calR$ and $\calO$ in its statement instantiated as follows.
\begin{enumerate}
    \item \emph{Instantiating the oracle $\calR$.} This is Grover search with marking function $r_1 < \cdot < r_2$. Note that Grover search does indeed return a uniformly random marked item as required to apply \Cref{lemma:randomized_search}.
    
    \item \emph{Instantiating the oracle $\calO$.} Let $x \in \mathbb{Z}^n$ and $u \in \mathbb{Z}$. Recall that $\LIS_{j,n}(x)$ decides whether $x$ contains a $j$-IS*. \Cref{tab:mt-lis} and \Cref{tab:mh-lis} show how an algorithm for computing $\LIS_{j,n}$ can be used to instantiate $\calO$ for computing $\mt_{j,n}$ and $\mh_{j,n}$, respectively. In the table headings, ``remove $i$'' is shorthand for ``turn into $*$ when $x$ is queried at position $i$''. (Note that being able to do this removal operation is the reason why we considered $k$-IS* instead of $k$-IS, and why we say that $k$-IS* is more susceptible to recursion than $k$-IS.)

\begin{table}[ht]
    \centering
    \begin{tabular}{c|c|c}
      \thead{Remove all $i\in [n]$ such that\ $x[i] > u$ \\ and compute $\LIS_{j,n}(x)$} & \thead{Remove all $i\in [n]$ such that $x[i] \geq u$ \\ and compute $\LIS_{j,n}(x)$} & Conclusion
      \\
      \hline
      $0$ & $0$ & $\mt_{j,n}(x)>u$
      \\
      $0$ & $1$ & Impossible
      \\
      $1$ & $0$ & $\mt_{j,n}(x)=u$
      \\
      $1$ & $1$ & $\mt_{j,n}(x)<u$
    \end{tabular}
\caption{\label{tab:mt-lis}Instantiating $\calO$ for computing $\mt_{j,n}$ using $\LIS_{j,n}$.}
\end{table}

\begin{table}[ht]
    \centering
    \begin{tabular}{c|c|c}
      \thead{Remove all $i\in [n]$ such that $x[i] < u$ \\ and compute $\LIS_{j,n}(x)$}  & \thead{Remove all $i\in [n]$ such that $x[i] \leq u$ \\ and compute $\LIS_{j,n}(x)$} & Conclusion
      \\
      \hline
      $0$ & $0$ & $\mh_{j,n}(x)<u$
      \\
      $0$ & $1$ & Impossible
      \\
      $1$ & $0$ & $\mh_{j,n}(x)=u$
      \\
      $1$ & $1$ & $\mh_{j,n}(x)>u$
    \end{tabular}
\caption{\label{tab:mh-lis}Instantiating $\calO$ for computing $\mh_{j,n}$ using $\LIS_{j,n}$.}
\end{table}

\end{enumerate}

\begin{Remarks*}\
\begin{enumerate}[topsep=4pt]
    \item Our randomized search generalizes the technique of quantum minimum finding \cite{minimum_finding_durr_hoyer_1996} because in quantum minimum finding, $s$ is additionally promised to be the minimum value in $S$ in \Cref{lemma:randomized_search}.
    \item \Cref{lemma:randomized_search} assumes that the oracles $\calO$ and $\calR$ are noiseless. In fact, the way we instantiate them means they can be erroneous with bounded probability. Naively, one might expect this to introduce an additional $\log\log(n)$ factor in the complexity of randomized search due to the need for error reduction since $\calO$ and $\calR$ are called $O(\log(n))$ times. However, techniques in \cite{noisybinary_feige_1994} allow us to remove this $\log\log(n)$ factor. 
\end{enumerate}
\end{Remarks*}

\begin{Lemma}[Randomized search]\label{lemma:randomized_search}
Let $S$ be a multi-set of $n$ integers and let $s\in S$. Let $\calO$ be an oracle that decides for any input $r\in \mathbb{Z}$ whether $s<r$, $s=r$, or $s>r$. Let $\calR$ be an oracle that, for any inputs $r_1,r_2 \in \mathbb{R}\cup\{+\infty,-\infty\}$, selects a uniformly random element $u$ from the multi-set $\{a\in S \mid r_1<a<r_2\}$. Then there exists a constant $c>0$ and a randomized classical algorithm that calls each of $\calO$ and $\calR$ at most $c \cdot \log(n/\delta)$ times to output $s$ with probability at least $1-\delta$.
\end{Lemma}

\begin{proof}
We show that the following randomized classical algorithm satisfies the requirements of the lemma. The commands in square brackets are only used in the analysis.

\begin{defaultenumerate}[]
    \item[] Set $r_1 = -\infty$ and $r_2 = \infty$. [Set $t=0$ and $S_0 \coloneqq S$.]
    
    \item[] \textbf{Repeat the following until success:}
    
    Use $\calR$ with input $r_1$ and $r_2$ to select a uniformly random element $u$ from $\{a\in S \mid r_1<a<r_2\}$.
    
    Use $\calO$ to decide whether $s<u$, $s=u$, or $s>u$.
    \begin{defaultenumerate}
        \item If $s=u$, output $s$.
        \item If $s<u$, set $r_2 = u$. [Set $S_t \coloneqq \{a \in S_{t-1} \mid a<u\}$.]
        \item If $s>u$, set $r_1 = u$. [Set $S_t \coloneqq \{a \in S_{t-1} \mid a>u\}$.]
    \end{defaultenumerate}
    
    [Set $t=t+1$.]
\end{defaultenumerate}

For $t\geq 0$, write $X_t \coloneqq |S_t|$, which is a random variable. Write $x_t \coloneqq \mathbb{E}[X_t]$. Initially, $S_0=S$ and $x_0 = n$.

Suppose that the rank of $s$ in $S_t$ is $p$ (if there are multiple $s$s, then $p$ can be taken as the rank of any). Then we have
\begin{equation}\label{eq:shrink}
    \expect[X_{t+1}|X_t] \leq \frac{1}{X_t} \Big(\sum_{i=1}^{p-1}(X_t-i) + \sum_{i = p}^{X_t-1} i\Bigr) \leq \frac{3}{4} X_t,
\end{equation}
where the first inequality is due to the possibility of $S$ containing duplicate elements (otherwise it would be equality) and the second inequality follows by considering the maximization of a quadratic polynomial in $p$. Taking the expectation of \Cref{eq:shrink} gives $x_{t+1}\leq \frac{3}{4}x_t$. Together with $x_0=n$, this implies $x_t\leq (3/4)^t n$. 

Therefore, by Markov's inequality, we have $\text{prob}(X_t > 1) \leq (3/4)^t n \leq \delta$ for $t \geq 10 \log(n/\delta)$. Therefore, since $s\in S_t$ for all $t\geq 0$ before the algorithm outputs $s$, if we force the algorithm to abort after $10 \log(n/\delta)$ steps, we will obtain $s$ at some point with probability at least $1-\delta$.
\end{proof}

\end{document}